\documentclass[journal]{IEEEtranTCOM}

\usepackage{amsmath, amsthm, amsfonts, amssymb, amsbsy}

\usepackage{graphicx}

\usepackage{pstricks}
\usepackage{algorithm}
\usepackage{multirow}
\usepackage{algorithmic}
\usepackage{bm}
\usepackage{cite}
\usepackage{url}
\usepackage{enumerate}
\usepackage{float}
\usepackage{amsmath}
\usepackage{amssymb}
\usepackage{mathrsfs}
\usepackage{lipsum} % 调用跨栏长公式宏包
\usepackage{amsmath,epsfig,amssymb,subfigure,bm,dsfont}
\usepackage{multicol} 
\usepackage{booktabs}
\usepackage{threeparttable}
\usepackage{amsmath}
\usepackage{threeparttable}

\newtheorem{theorem}{Theorem}

\newtheorem{definition}{Definition}

\newtheorem{corollary}{Corollary}
\newtheorem{lemma}{Lemma}
\newtheorem{example}{Example}

\newtheorem{remark}[theorem]{Remark}

\newcommand{\twotriangle}{\hfill $\bigtriangleup \bigtriangleup$  }
\newcommand{\eax}{\twotriangle  \end{example}}
\newcommand\bim{\begin{itemize}}
\newcommand\eim{\end{itemize}}

\begin{document}

\title{
An Integrated Sensing and Communications System Based on Affine Frequency Division Multiplexing 
}

\author{
	Yuanhan Ni, {\it Member, IEEE}, Peng Yuan, Qin Huang, {\it Senior Member, IEEE}, Fan Liu, {\it Senior Member, IEEE}, Zulin Wang, {\it Member, IEEE}
	\thanks{This work was supported in part by the China Postdoctoral Science Foundation under Grant Number 2024M764088, in part by the National Natural Science Foundation of China under Grant 61971025, 62331002,	62071026 and 62331023. Part of this paper has been accepted by the 2022 International Symposium on Wireless Communication Systems (ISWCS)\cite{ni2022afdm}. (corresponding author: Zulin Wang and Qin Huang.)}
	\thanks{Y. Ni, P. Yuan, Q. Huang and Z. Wang are with the School of Electronic and Information Engineering, Beihang University, Beijing 100191, China (e-mail: yuanhanni@buaa.edu.cn; yuanpeng9208@buaa.edu.cn; qhuang.smash@gmail.com; wzulin@buaa.edu.cn).}
	\thanks{F. Liu is with the National Mobile Communications Research Laboratory, Southeast University, Nanjing 210096, China (e-mail: f.liu@ieee.org).}
}

\maketitle

\begin{abstract}
This paper proposes an integrated sensing and communications (ISAC) system based on affine frequency division multiplexing (AFDM) waveform. To this end, a metric set is designed according to not only the maximum tolerable delay/Doppler, but also the weighted spectral efficiency as well as the outage/error probability of sensing and communications. This enables the analytical investigation of the performance trade-offs of AFDM-ISAC system using the derived analytical relation among metrics and AFDM waveform parameters. Moreover, by revealing that delay and the integral/fractional parts of normalized Doppler can be decoupled in the affine Fourier transform-Doppler domain, an efficient estimation method is proposed for our AFDM-ISAC system, whose unambiguous Doppler can break through the limitation of subcarrier spacing. Theoretical analyses and numerical results verify that our proposed AFDM-ISAC system may significantly enlarge unambiguous delay/Doppler while possessing good spectral efficiency and peak-to-sidelobe level ratio in high-mobility scenarios.
\end{abstract}


\section{Introduction}

Next-generation wireless systems (beyond 5G/6G) are expected to significantly improve spectral and energy efficiencies, support ubiquitous connections of everything and maintain reliable communications in high-mobility scenarios\cite{ni2022afdm,liu2020joint}. The integrated sensing and communications (ISAC) technique is one of the critical enablers of beyond 5G/6G due to its ability to simultaneously improve spectral and energy efficiencies and extract information of the environment \cite{liu2020joint}.

On the one hand, ISAC waveform design and corresponding signal processing play vital roles in the ISAC system. Multicarrier waveforms have been widely studied as ISAC waveforms due to their advantages of high communication spectral efficiency, robustness against multipath fading, and good ambiguity characteristics, etc. For example, the orthogonal frequency division multiplexing (OFDM) waveform has been used to simultaneously realize sensing and communications (S$\&$C) functionalities\cite{sturm2011waveform,zeng2020joint}. To this end, a symbol-wise division-based method has been proposed to suppress the sidelobes of radar images caused by random communications symbols and estimate the parameters of targets\cite{sturm2011waveform}. The unambiguous Doppler of this method is limited by the subcarrier spacing. The symbol division operation has been replaced with the symbol conjugate multiplication operation by utilizing the cyclic cross-correlation (CCC) to avoid amplifying the noise background when symbols have non-constant modulus in the frequency domain \cite{zeng2020joint}. However, for the OFDM-ISAC system, the bit error rate (BER) of communications and the peak-to-sidelobe level ratio (PSLR) of radar images may deteriorate severely in high-mobility scenarios\cite{Raviteja2018Interference,bemani2023affine}. 

To improve the performance of ISAC systems in high-mobility scenarios, orthogonal time-frequency space (OTFS)-based ISAC waveforms have been investigated  \cite{Raviteja2018Interference,Gaudio2020On}. OTFS waveform spreads information symbols in the delay-Doppler (DD) domain\cite{hadani2017orthogonal,Raviteja2018Interference}. Consequently, OTFS can deal with significant Doppler shifts and obtain both time and frequency diversities in doubly selective channels. Based on this, the OTFS-ISAC system has been considered, where an efficient algorithm has been proposed to estimate range and velocity\cite{Gaudio2020On}. It has been shown that OTFS can achieve similar sensing performance with OFDM and FMCW waveforms. 

{ To couple the communications channel and the modulated signal in the DD domain, a promising and practical delay-Doppler multi-carrier modulation, namely orthogonal delay-Doppler division multiplexing (ODDM), has been proposed\cite{lin2022orthogonal}. ODDM adopts the delay-Doppler domain orthogonal pulses (DDOP) and thus guarantees the orthogonality with respect to fine DD resolutions. The results in \cite{lin2022orthogonal} have verified that ODDM has advantages in out-of-band emission and BER compared with the traditional OTFS. Moreover, ODDM-ISAC systems have been studied to integrate the advantage of ODDM in the ISAC systems\cite{wang2024exploring,zhang2024multiuser}. Specifically, the estimated channel in the sensing process is used to improve the BER performance in the ODDM communication process, assuming that the sensing and communications channels are identical\cite{wang2024exploring}.} 

A new multicarrier waveform using chirp orthonormal basis, namely orthogonal chirp-division multiplexing (OCDM), has been proposed based on the discrete Fresnel transform (FrT)\cite{Ouyang2016Orthogonal}.
OCDM multiplexes a set of orthogonal chirps that are complex exponentials with linearly varying instantaneous frequencies. As each information symbol occupies the entire bandwidth, OCDM can achieve full diversity in frequency-selective channels, outperforming OFDM. 
It has been shown that the sidelobe level of the resulting radar image of the OCDM-ISAC system is slightly increased compared with the OFDM-ISAC system\cite{Oliveira2020An}. However, OCDM can only achieve partial diversity in doubly selective channels\cite{bemani2023affine}.    

Recently, another chirp multicarrier waveform, namely affine frequency division multiplexing (AFDM), has been proposed for communications systems by multiplexing information symbols in the affine Fourier transform (AFT) domain \cite{bemani2023affine}. AFDM waveform is able to separate all paths in the AFT domain by optimizing its parameters. As a result, AFDM can achieve full diversity in doubly selective channels. Moreover, the pilot-aided channel estimation and equalization algorithms in the AFT domain have been proposed for AFDM\cite{bemani2023affine,yin2022pilot}. It has been shown that compared with OTFS, AFDM has comparable communications performance in terms of BER but with lower complexity and the advantage of less channel pilot overhead\cite{bemani2023affine}. Then, the channel estimation scheme has been extended from the single antenna scenario to the MIMO scenario\cite{yin2022pilot}. An AFDM empowered sparse code multiple access (SCMA) system, referred to as AFDM-SCMA, has been proposed with significantly improved spectrum efficiency for massive connectivity in high mobility channels\cite{luo2024afdm}. The results in \cite{luo2024afdm} have shown that the proposed AFDM-SCMA significantly outperforms OFDM-SCMA in both uncoded and coded systems.
Moreover, benefitting from the fact that AFT is the genericized form of FrT and Fourier transform, AFDM waveform can be backwards compatible with traditional OFDM and chirp waveforms. 
Accordingly, it necessitates investigating the ISAC system based on AFDM waveform and the corresponding efficient method for estimating sensing parameters, especially in high-mobility scenarios\cite{ni2022afdm,bemani2024integrated,zhu2023low,zhu2024afdm}.

On the other hand, the performance metrics are the cornerstone of analyzing trade-offs and optimizing waveforms of ISAC systems.
For communications systems, the achievable communications spectral efficiency (CSE), the maximum tolerable delay/Doppler, and BER have been regarded as metrics to measure the efficiency, the operating range and the reliability of the system\cite{shitomi2022spectral}. Here, the operating range refers to the range of delay/Doppler within which the system can work reliably. 
%In other words, the reliability of the system will be seriously deteriorated if the delay/Doppler of channel exceeds the operating range of the system. 
The unambiguous range/velocity, the probability of detection, and the estimation accuracy have been used to measure the operating range and the reliability performance of sensing systems\cite{mark2010principles}. However, the above traditional metrics of S$\&$C systems usually have different structures and physical meanings, which makes it challenging to analyze the performance trade-off of S$\&$C and optimize ISAC waveforms. 

To this end, new metrics with similar structures and/or physical meanings for S$\&$C functionalities are researched mainly motivated by the information theory perspective\cite{liu2022integrated}. For example, radar capacity has been proposed to evaluate how much information of targets is attainable by a moving-target indication radar after performing a periodic search\cite{guerci2015joint}. The radar estimation rate has been introduced to measure how much uncertainty in the target is cancelled after estimating the delay or Doppler of the target\cite{chiriyath2015inner}. On top of that, an inner bound on radar estimation rate has been derived. The Cramér-Rao bound (CRB)-rate region has been used to characterize the fundamental trade-off between S$\&$C\cite{xiong2023fundamental}. An outer bound and various inner bounds on CRB-rate regions have been proposed. The conditional mutual information has been used as the unified performance metric for both S$\&$C, where traditional sensing metrics are connected with the sensing mutual information (MI) from a rate-distortion perspective\cite{dong2023rethinking}.
The trade-off between the ambiguity function of sensing and the achievable information rate of communications has been investigated by optimizing the input distribution utilizing the probabilistic constellation shaping method\cite{du2023reshaping}.  
{ The performances of the multi-cell ISAC network have been analyzed using the coverage and the ergodic rate as the unified performance metrics of S$\&$C\cite{gan2024coverage}.} 
Recently, a theoretical sensing rate has been introduced to measure the information obtained by a pulse-Doppler radar system\cite{choi2023information}. While there are initial efforts towards defining the theoretical ``capacity'' of sensing, it still necessitates investigating the achievable efficiency of sensing functionality for specific ISAC waveforms.

In this paper, we propose an AFDM-ISAC system, where the AFDM waveform is regarded as the ISAC waveform to realize S$\&$C functionalities simultaneously. For that, two new metrics, namely sensing spectral efficiency (SSE) and sensing outage probability (SOP) are introduced. Based on this, a metric set is designed according to the weighted CSE and SSE, the maximum tolerable delay/Doppler, and the weighted BER and SOP. Then, we analytically investigate the performance trade-offs of AFDM-ISAC systems using the derived analytical relation among metrics and AFDM waveform parameters. An efficient estimation method is proposed to extract the delay and the integral/fractional parts of normalized Doppler. Numerical results verify that the proposed AFDM-ISAC system can significantly enlarge maximum tolerable delay/Doppler and possess good spectral efficiency and PSLR in high-mobility scenarios, compared with the OFDM-ISAC system.
For clarity, we summarize our contributions as follows:
\begin{itemize}
	{ \item  We propose an AFDM-ISAC system. For that, we introduce SSE and SOP that can evaluate sensing efficiency and reliability (including detection and estimation processing) and then propose a metric set for ISAC systems. SSE and SOP have structures and physical meanings similar to the existing CSE and BER of communications. It is revealed that the traditional sensing performances (e.g., resolution, unambiguous delay/Doppler, probability of detection, probability of false alarm, and estimation error) can be reflected in SSE and SOP. Moreover, the trade-off between SSE and SOP for sensing is analogous to the trade-off between CSE and BER for communications.
	
	\item We derive the relation between the proposed metric set and AFDM parameters and analytically investigate the performance trade-offs of the AFDM-ISAC system. The derived relation and performance trade-offs reveal that the performances of AFDM-ISAC system are strongly correlated with AFDM parameters $c_1$, $N$ and $N_{cp}$. We also analytically demonstrate that the maximum tolerable Doppler of the proposed AFDM-ISAC system can be more than five times that of the OFDM-ISAC system with appropriate parameters. Moreover, the derived relation allows us to provide a design guideline on selecting AFDM parameters.
	
	\item We propose an efficient estimation method for our AFDM-ISAC system to estimate the delay and the Doppler of sensing targets in the AFT-Doppler domain. The range of estimated Doppler can break through the limitations of subcarrier spacing with an acceptable complexity.}	
\end{itemize} 

The rest of this paper is organized as follows. Section II introduces the system model. In Section III, we propose an AFDM-ISAC system and design a metric set. Performance analyses of AFDM-ISAC system are proposed in Section IV. Section V designs an estimation method of sensing parameters for the AFDM-ISAC system. Numerical results are shown in Section VI. Section VII concludes the paper.

\emph{Notation:} Throughout the paper, $\mathbf{X}$, $\mathbf{x}$, and $x$ denote a matrix, vector, and scalar, respectively.
$\left\lfloor \cdot \right\rfloor$, $\left\lfloor \cdot \right\rceil$ and $\left\lceil \cdot \right\rceil$ denote the floor, round and cell functions, respectively. ${\left\langle \cdot \right\rangle _N}$, $\delta\left[\cdot\right]$, $\odot$, $\left( \cdot \right)^{*}$ and $\left( \cdot \right)^{\rm H}$ are the modulo $N$ operation, the Dirac delta function, the Hadamard products, the conjugate operation, and the Hermitian transpose, respectively. $\mathrm{diag}\left(\mathbf{x}\right)$ returns a diagonal matrix with the element of $\mathbf{x}$ on the main diagonal.

\section{Preliminaries}

\subsection{Basic Concepts of AFDM}

We briefly review the basic concepts of AFDM proposed in \cite{bemani2023affine}. Let $\mathbf{x}$ denote an $N {\times} 1$ vector of quadrature amplitude modulation (QAM) symbols. To estimate the channel state information, $N_{pg}$ pilot and guard symbols should be inserted into $\mathbf{x}$\cite{bemani2023affine}. Hence, the length of valued data symbol vector $\mathbf{d}$ is given by $M = N - N_{pg}$. The $N$ points inverse discrete affine Fourier transform (IDAFT) is performed to map $\mathbf{x}$ from the AFT domain to the delay domain, i.e.,\cite{bemani2023affine}
\begin{equation}
	\mathbf{s}\left[ {n} \right] = \frac{1}{{\sqrt N }}\sum\nolimits_{m = 0}^{N - 1} {\mathbf{x}\left[ {m} \right]} {e^{j2\pi \left( {{c_1}{n^2} + \frac{1}{N}mn + {c_2}{m^2}} \right)}} ,
\end{equation}
where $c_1$ and $c_2$ are the AFDM parameters, and $n = 0, \ldots ,N {-} 1$. Then, a chirp-periodic prefix (CPP) with a length of $N_{cp}$ is added,
which is defined as\cite{bemani2023affine}
\begin{equation}\label{eq:symbol_AFDM}
\mathbf{s}\left[ {n} \right] {=} \mathbf{s}\left[ {n {+} N} \right]{e^{ - i2\pi {c_1}\left( {{N^2} + 2Nn} \right)}},n =  - {N_{cp}}, \ldots , - 1.
\end{equation}

Considering a communications channel with $P$ paths, the gain
coefficient, time delay and Doppler shift of the $i$-th path are denoted by ${h_i}$, ${\tau _i}$, ${f_{d,i}}$, respectively.
After transmission over the channel, the received signal vector in the time domain is given by \cite[Eq. (6)]{wu2022integrating}
\begin{equation}\label{eq:received_sig_time}
	%\mathbf{\tilde r}\left[ n \right] = \sum\nolimits_{i = 1}^P {{{\tilde h}_i}} \mathbf{\tilde s}\left[ {n - {l_i}} \right]{e^{j2\pi {f_i}n}} + {\mathbf{\tilde w}_r}\left[ n \right],
	\mathbf{r}\left[ n \right] = \sum\nolimits_{i = 1}^P {{{\tilde h}_i}} \mathbf{s}\left[ {n - {l_i}} \right]{e^{j2\pi {f_i}n}} + {\mathbf{w}}_\mathrm{t}\left[ n \right],
\end{equation}
where ${{\tilde h}_i} {=} {h_i}{e^{ - j2\pi {f_{d,i}}{\tau _i}}}$, $n \in \left[ { - {N_{cp}},N - 1 } \right]$, $\mathbf{{w}}_\mathrm{t}\sim \mathcal {CN}\left( {0\text{,} \sigma_{c} ^2\mathbf{I}} \right)$ is an additive Gaussian noise vector in the time domain, ${l_i} {=} {{{\tau _i}} \mathord{\left/
		{\vphantom {{{\tau _i}} {{t_s}}}} \right.
		\kern-\nulldelimiterspace} {{t_s}}}$, ${f_i} {=} {f_{d,i}}{t_s}$, and ${t_s}$ denotes the sampling interval. 
 After discarding CPP and performing $N$ points discrete affine Fourier transform (DAFT), the input-output relationship in the AFT domain can be written in the matrix form as\cite{bemani2023affine}
\begin{equation}
	\mathbf{y} = {\mathbf{H}_\mathrm{eff}}\mathbf{x} + \mathbf{{ w}}_\mathrm{a}= \sum\nolimits_{i = 1}^P {{\tilde h_i}{\mathbf{H}_{\mathrm{A},i}}\mathbf{x}} + \mathbf{{ w}}_\mathrm{a},
\end{equation}
where ${\mathbf{H}_\mathrm{eff}} {=} \mathbf{A}{\mathbf{H}_\mathrm{c,t}}\mathbf{A}^{\rm{H}}$ with ${\mathbf{H}_\mathrm{c,t}} = \sum\nolimits_{i = 1}^P {{\tilde h_i}{\mathbf{\Gamma} _i}{\mathbf{\Delta} _{{f_i}}}{\mathbf{\Pi} ^{\left({l_i}\right)}}}$ being the communications channel matrix in the time domain, $\mathbf{\Pi}$ is the forward cyclic-shift matrix, ${\mathbf{\Gamma} _i} = \mathrm{diag}\left( {\left\{ {\begin{array}{*{20}{c}}
			{{e^{ - i2\pi {c_1}\left( {{N^2} - 2N\left( {{l_i} - n} \right)} \right)}}},&{n < {l_i}},\\
			1,&{n \ge {l_i}}.
	\end{array}} \right.} \right)$, ${\mathbf{\Delta} _{{f_i}}} = \mathrm{diag}\left( {{e^{ i2\pi {f_i}n}},n \in \left[0, {N-1}\right]} \right)$, ${\mathbf{H}_{\mathrm{A},i}}=\mathbf{A}{\mathbf{\Gamma} _i}{\mathbf{\Delta} _{{f_i}}}{\mathbf{\Pi} ^{\left({l_i}\right)}}\mathbf{A}^{\rm{H}}$, $\mathbf{A} = {\mathbf{\Lambda} _{{c_2}}}\mathbf{F}{\mathbf{\Lambda} _{{c_1}}}$,  $\mathbf{F}$ is the discrete Fourier transform (DFT) matrix, ${\mathbf{\Lambda} _{{c_i}}} {=}  \mathrm{diag}\left( {{e^{ - j2\pi c_i{n^2}}},n = 0, \ldots ,N {-} 1}, i=1,2 \right)$, and $\mathbf{{ w}}_\mathrm{a}=\mathbf{A}\mathbf{w}_\mathrm{t}$. ${\mathbf{H}_{\mathrm{A},i}}\left[ {p,q} \right]$ is given by\cite{bemani2023affine}
\begin{eqnarray}
	{\mathbf{H}_{\mathrm{A},i}}\left[ {p,q} \right] = \frac{1}{N}{e^{j\frac{{2\pi }}{N}\left( {N{c_1}l_i^2 - ql_i + N{c_2}\left( {{q^2} - {p^2}} \right)} \right)}}\mathcal{F}_i\left[ {p,q} \right],
\end{eqnarray}
where $p$ and $q \in \left[ 0,N-1 \right]$, ${\mathcal{F}_i}\left[ {p,q} \right] {=} \frac{{{e^{ - j2\pi \left( {p - q - {\nu_i} + 2N{c_1}{l_i}} \right)}} - 1}}{{{e^{ - j\frac{{2\pi }}{N}\left( {p - q - {\nu_i} + 2N{c_1}{l_i}} \right)}} - 1}}$ with ${\nu _i} = N{f_i} = \frac{{{f_{d,i}}}}{{\Delta f}} = {\alpha _i} + {a_i} \in \left[ { - {\nu _{\max }}, {\nu _{\max }}} \right]$ being the Doppler shift normalized with respect to the 
subcarrier spacing ${\Delta f}$, and ${\alpha _i} {\in} \left[ { - {\alpha _{\max }},{\alpha _{\max }}} \right]$ and ${a_i} {\in} \left( { - \frac{1}{2},\frac{1}{2}} \right]$ represent the integral and fractional part of ${\nu _i}$, respectively. 

In the integral normalized Doppler shift case, i.e., $a_i{=}0$, there is only one non-zero element in each row of $\mathbf{H}_{\mathrm{A},i}$, i.e., 
\begin{equation}
	{\mathbf{H}_{\mathrm{A},i}}\left[ {p,q} \right] = \left\{ {\begin{array}{*{20}{c}}
			\hspace{-0.2cm}{{e^{j\frac{{2\pi }}{N}\left( {N{c_1}l_i^2 - ql_i + N{c_2}\left( {{q^2} - {p^2}} \right)} \right)}}},\hspace{-1.5ex}&{q = {{\left\langle {p + lo{c_i}} \right\rangle}_N}},\\
			0,&{otherwise},
	\end{array}} \right. 
\end{equation}
where $lo{c_i} {=} {{{\left\langle {2N{c_1}{l_i} {-} {\alpha _i}} \right\rangle}_N}}$. Hence, the input-output relation in the AFT domain can be rewritten as 
\begin{align}\label{eq:in_out_relation_2}
	&\mathbf{y}\left[ {p} \right] \hspace{-0.5ex} = \hspace{-0.5ex}\sum\limits_{i = 1}^P {{\tilde h}_i}
	{e^{j\frac{{2\pi }}{N}\left( N{c_1}l_i^2 {- q_il_i + N{c_2}\left( {{{q}_i^2} - {p^2}} \right)} \right)}}\mathbf{x}\left[ {q_i} \right] \hspace{-0.5ex}{+} \mathbf{ w}_\mathrm{a}\left[ {p} \right],
\end{align}
where ${q_i {=} {{\left\langle {p + lo{c_i}} \right\rangle}_N}}$ and $p\in \left[0,{N-1}\right]$.
In the fractional normalized Doppler shift case, 
there are $2k_v {+} 1$ non-zero elements and the peak is still at ${q {=} {{\left\langle {p + lo{c_i}} \right\rangle}_N}}$ in each row of $\mathbf{H}_{\mathrm{A},i}$\cite{bemani2023affine}. 

\section{AFDM-ISAC System and A Metric Set}

{ This section proposes an AFDM-ISAC system. Specifically, an AFDM-based generic ISAC waveform is first shown. Then, a metric set is designed for the AFDM-ISAC system, in which we introduce two new metrics, namely SSE and SOP, to evaluate the efficiency and the reliability of sensing, respectively.}

\subsection{AFDM-Based Generic ISAC Waveform}\label{AFDM_waveform}

We consider a monostatic AFDM-ISAC setup, where the ISAC base station (BS) transmits the AFDM-based ISAC waveform to downlink users and simultaneously receives echoes reflected by targets around it, as shown in Fig. \ref{fg:AFDM_based_ISAC}.
\begin{figure}
	\centering	
	\includegraphics[width=2.0in]{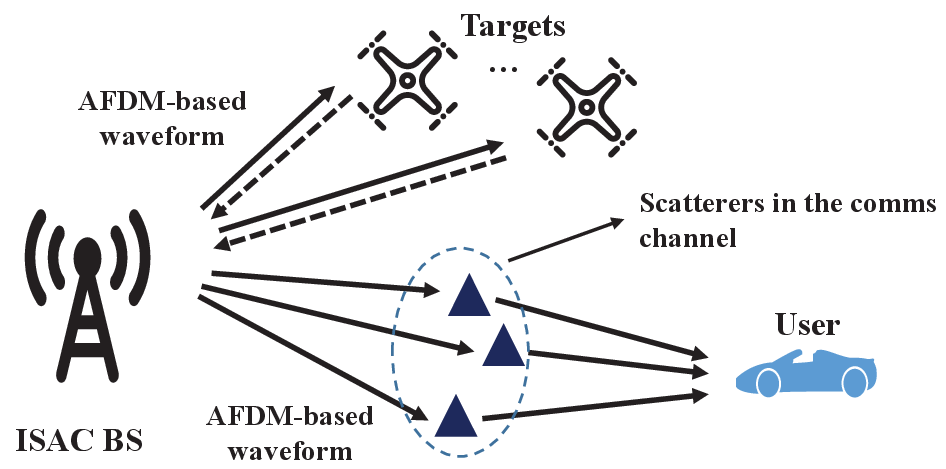}
	\caption{The ISAC scenario using AFDM waveform.
		\label{fg:AFDM_based_ISAC}}
		\vspace*{-5pt} %留空白，可自己调整
\end{figure}

The framework of proposed AFDM-based generic ISAC waveform is shown in Fig. \ref{fg:generic_waveform_framework}, where $N_{sym}$ AFDM symbols are packed as an ISAC waveform (frame) $\mathbf{{s}}_\mathrm{I} \in {\mathbb{C}^{{N_sN_{sym}} \times 1}}$ in the time domain, which is used to carry $N_{sym}Mlog_2\left(M_{mod}\right)$ communications bits and perform a sensing (detection and estimation) process, where $N_s=N+N_{cp}$. 
The proposed generic ISAC waveform has configurable parameters $M_{mod}$, $M$, $N$, $c_1$, $c_2$, $N_{cp}$, and $N_{sym}$.  
By configuring these parameters, the proposed ISAC waveform can not only generate traditional single-carrier ($c_1 = c_2 = 0$, $N=1$), OFDM ($c_1 = c_2 = 0$), OCDM ($c_1 = c_2 = \frac{1}{N}$), and LFM ($c_2 = 0$, $M=1$) ISAC waveforms, but also achieve performance trade-off between S$\&$C.
Based on the previous definition of AFDM symbol in (\ref{eq:symbol_AFDM}), the expression of AFDM-based genaral ISAC waveform $\mathbf{s}_I$ can be written as \begin{align}\label{eq:s_n3}
	&\mathbf{s}_\mathrm{I}\left[ n \right]=  \frac{1}{{\sqrt N }}\sum\limits_{k = 0}^{N_{sym} - 1}\sum\limits_{m = 0}^{N - 1} {{\mathbf{X}\left[ m,k \right]}} g\left({n}-kN_s\right)  \nonumber \\
	& \quad \quad \quad \quad \times {e^{j2\pi \left( {{c_1}{\left({n}-kN_s\right)^2}{\rm{ + }}{c_2}{m^2} + \left({n}-kN_s\right)m/N} \right)}},
\end{align}
where $\mathbf{X} =\left[\mathbf{x}_0, \cdots, \mathbf{x}_{N_{sym}-1} \right]\in {\mathbb{C}^{{N} \times N_{sym}}}$ denotes the communications symbols block, 
$g\left( n \right) = \left\{ {\begin{array}{*{20}{c}}
		{1,}&{n \in \left[ { - {N_{cp}},N - 1} \right],}\\
		{0,}&{otherwise.}
\end{array}} \right.$ being the window function, and $n \in \left[ { - {N_{cp}},N - 1 + \left( {{N_{sym}} - 1} \right){N_s}} \right]$.

\begin{figure}[htbp] 	
	\centerline{\includegraphics[width=3.5 in]{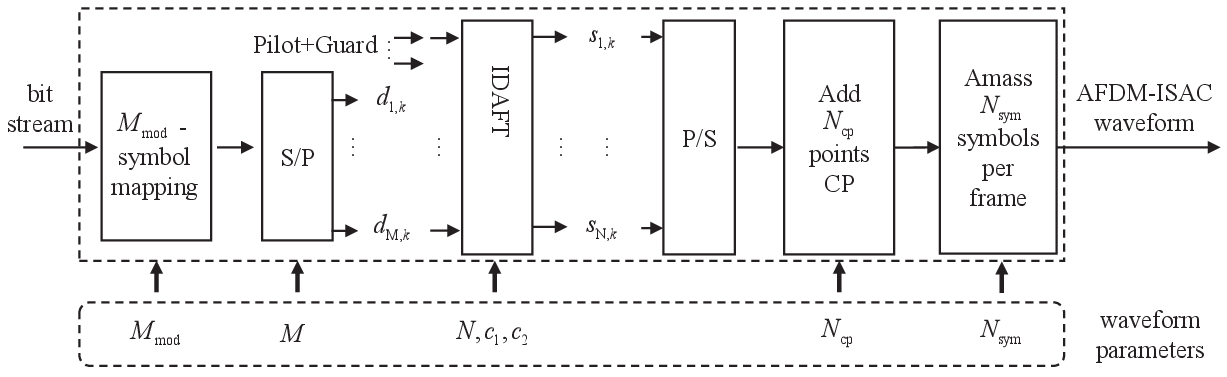}}
	\caption{The proposed AFDM-based generic ISAC waveform.}
	\label{fg:generic_waveform_framework}	
	\vspace*{-5pt} %留空白，可自己调整
\end{figure}

\subsection{A Metric Set for AFDM-ISAC System} \label{sec_metric}

{ In this subsection, we commence by briefly reviewing the performance metrics of communications functionality. Then, analogous to communications, we introduce two new metrics, namely SSE and SOP, to evaluate the efficiency and the reliability of sensing functionality. Finally, a metric set for the ISAC system is proposed from the efficiency, operating range and reliability perspectives to comprehensively evaluate the performances of S$\&$C. Since there are many variables in this subsection, for clarity, the main variables in this subsection are listed in Table \ref{tab:variable_tab}.
}

{
\begin{table}
	\centering
	%%\caption{****}
	\caption{{Main Variables in Section III-B}}\label{tab:variable_tab}
	\renewcommand\arraystretch{1.3}
	\begin{threeparttable}
		\begin{tabular}{|c|l|}
			\hline
			\textbf{{Variable}}                       &\textbf{{Meaning}}                       \\ \hline
			$R_{\max}$/$V_{\max}$ &  {The maximum unambiguous range/velocity}    \\ \hline
			$\tau _{\max}$/$f_{d,\max}$ &  {The maximum tolerable delay/Doppler}     \\ \hline
			$U^{n,k} \in \left\{ 0,1 \right\}$ &  {Absent or present of the target in the $\left(n,k\right)$-th cell}     \\ \hline			
			${{\tau }^{n,k}}$/${{f}_{d}^{n,k}}$ &  {Actual delay/Doppler of the target in the $(n,k)$-th cell}     \\ \hline
			${{\bar \tau }^{n,k}}$/${{\bar f}_{d}^{n,k}}$ &  {Delay/Doppler of the centre of the $(n,k)$-th cell}   \\ \hline
			$\tilde \tau ^{n,k}$/$\tilde f_{d}^{n,k}$ &  {Observed delay/Doppler of the target in the noisy case}     \\ \hline			
			${\hat U}^{n,k}$/${{\hat \tau }^{n,k}}$/${{\hat f}_{d}^{n,k}}$ &  {Discretely estimated values of $U^{n,k}$/${{\tau }^{n,k}}$/${{f}_{d}^{n,k}}$   }  \\ \hline			
			$\sigma _{\tau}^{n,k} $/$\sigma _{f_d}^{n,k}$ &  {Lower bound of delay/Doppler estimation}     \\ \hline			
			$D_{\tau }^{n,k}$/$D_{f_d }^{n,k}$ &  {Interval of delay/Doppler of sub-cell of $(n,k)$-th cell }    \\ \hline			
			$\bar D_{R}$ / $\bar D_{V}$ &  {Given requirement of estimation error of range/velocity }    \\ \hline			
			$\bar D_{\tau }$/$\bar D_{f_d }$ &  {Given requirement of estimation error of delay/Doppler }    \\ \hline
			$P_{D}^{n,k}$/$P_{fa}$ &  {Probability of detection/false alarm in the $(n,k)$-th cell }   \\ \hline	
			${\eta _{sen}}$/${P_{e,sen}}$ &  {Defined SSE/SOP}     \\ \hline
			$P_{e,sen}^{\min }$/$P_{e,sen}^{\max }$ &  {Lower/Upper bound of SOP }    \\ \hline
		\end{tabular}
	\end{threeparttable}
\end{table}
}

Firstly, we review the performance metrics of communications. CSE can be used to measure the efficiency of communications, which evaluates practically transmitted information bit per time and frequency resources for a given waveform \cite{feher19871024,shitomi2022spectral}. 
For a given waveform with the time interval $T_A$ and bandwidth $B$, the maximum number of valued data symbols $N_{data}$ is given by $N_{data} = MN_{sym}$ according to \ref{AFDM_waveform}. Assume that each data symbol is chosen from the $M_{mod}$-QAM constellation. Therefore, the achievable CSE of the waveform is expressed as
\begin{eqnarray}\label{eq:CSE}
	\eta _{com} = \frac{{1}}{{B T_A}}N_{data}{\log _2}\left( M_{mod} \right).
\end{eqnarray}

The maximum tolerable delay/Doppler of channel and the delay/Doppler resulotion of channel can be regarded as the metrics to measure the operating range of communications. Specifically, the maximum tolerable delay/Doppler measures the maximum channel delay/Doppler that the waveform can tolerate, determining the upper bound on the operating range. The delay/Doppler resolution measures the ability to resolve the smallest channel delay/Doppler, representing the lower bound on the operating range. 

Due to the presence of noise, the received symbols may contain errors. BER can be used to measure the reliability of communications. 
If the receiver wants to achieve the CSE with an acceptable BER, it needs sufficient communication signal-to-interference-plus-noise ratio (SINR) as a guarantee. 

Motivated by communication metrics, we next define two new metrics, i.e., SSE and SOP, to measure the efficiency and the reliability of sensing for the given waveform, which have similar structures and physical meanings to the existing CSE and BER of communications, respectively.

\vspace{1ex}
\noindent\emph{(1) Proposed SSE}
\vspace{1ex}

{ We first define SSE to measure the efficiency of sensing functionality.}
Different from communication, the sensing transmitter sends a waveform that does not carry any information about sensing targets. After being scattered by targets, the echo carries continuous target parameter information. The sensing receiver detects whether the target exists and then estimates parameters to obtain information about sensing targets. 

For a given waveform with the time interval $T_A$ and bandwidth $B$, the maximum number of detectable targets is equal to the number of delay-Doppler resolution cells, which is given by\cite{mark2010principles}
\begin{eqnarray}\label{eq:N_cell}
	N_{cell} = \frac{{R_{\max }}}{{{\Delta _R}}}\frac{{2{v_{\max }}}}{{{\Delta _V}}} = \frac{{\tau _{\max }}}{{{\Delta _ \tau}}}\frac{{2{f_{d,\max }}}}{{{\Delta _{f_d}}}} = N_{\tau}N_{f_d} {,}
\end{eqnarray}
where $N_{\tau} = {{{\tau _{\max }}} \mathord{\left/
		{\vphantom {{{\tau _{\max }}} {{\Delta _\tau }}}} \right.
		\kern-\nulldelimiterspace} {{\Delta _\tau }}}$, and $N_{f_d} = {{2{f_{d,\max }}} \mathord{\left/
		{\vphantom {{2{f_{d,\max }}} {{\Delta _{{f_d}}}}}} \right.
		\kern-\nulldelimiterspace} {{\Delta _{{f_d}}}}}$.

Following \cite{choi2023information}, we define a binary random variable $U^{n,k} \in \left\{ 0,1 \right\}$ with Bernoulli parameter $\gamma$ to denote the target presence at the $(n, k)$ delay-Doppler cell, i.e., ${\Pr}\left(U^{n,k}=1\right)=\gamma$. 
When the Neyman-Pearson (NP) detector is adopted, the obtained target information by detecting whether a target is present in the $(n, k)$ cell can be expressed as\cite{choi2023information}
\begin{align}\label{eq:info_det}
%	&= H_b\left( \gamma  \right) - {H_{NP}}\left( {\gamma ,{P_{fa}},{P_{D}^{n,k}}} \right)  \nonumber \\
	&I\left( U^{n,k}; {\hat U}^{n,k} \right) = H_b\left( {\left( {1 - \gamma } \right){P_{fa}} + \gamma {P_{D}^{n,k}}} \right) \nonumber \\ 
	& \ \ \ \ \ \ \ \ \ \ \ \ \ \  - \left( {1 - \gamma } \right)H_b\left( {{P_{fa}}} \right) - \gamma H_b\left( {1 - {P_{D}^{n,k}}} \right),
\end{align}
where $H_b\left( p \right) =  - p {\log _2}p - \left( {1 - p} \right){\log _2}\left( {1 - p} \right)$ is a binary entropy function with parameter $p \in [0,1]$. $P_{fa}={\Pr}\left({\hat U}^{n,k}=1|U^{n,k}=0\right)$ and $P_{D}^{n,k}={\Pr}\left({\hat U}^{n,k}=1|U^{n,k}=1\right)$. 
$P_{D}^{n,k}$ may be different in each cell due to different SINR. Moreover, according to Remark 1 in \cite{choi2023information}, the parameter $\gamma$ evaluates the sparsity of targets in the delay-Doppler cells, and it can choose $\gamma = 0.5$, if the sensing system does not have any information about the environment.

Once a target is detected, the sensing system may focus on estimating the target parameters. This paper mainly concentrates on estimating the range (delay) and velocity (Doppler) of the target.
When a target is present in the $(n,k)$-th cell, it can be assumed that both $\tau^{n,k}$ and $f_{d}^{n,k}$ follow uniform distribution in the $(n,k)$-th cell, if the sensing system does not have any additional information about the target\cite[p. 472]{bar1995multitarget}. Hence, the conditional probability density functions (pdfs) of delay and Doppler can be given by $f\left(\tau^{n,k}|U^{n,k}=1\right) = \frac{1}{\Delta _{\tau}}, \tau ^{n,k} \in \left[{\bar \tau }^{n,k}-0.5\Delta _{\tau}, {\bar \tau }^{n,k}+0.5\Delta _{\tau}\right)$ and $f\left(f_{d}^{n,k}|U^{n,k}=1\right) = \frac{1}{\Delta _{f_d}}, f_{d}^{n,k} \in \left[{\bar f }_{d}^{n,k}-0.5\Delta _{f_d}, {\bar f }_{d}^{n,k}+0.5\Delta _{f_d}\right)$. Following \cite[Eq. 13]{choi2023information}, when the target is not present in the $(n,k)$-th cell, we model pdfs of delay and Doppler as $f\left(\tau^{n,k}|U^{n,k}=0\right) =\delta\left(\tau ^{n,k}\right)$ and $f\left(f_{d}^{n,k}|U^{n,k}=0\right) =\delta\left(f _{d}^{n,k}\right)$, respectively. Hence, the pdfs of delay and Doppler are given by

\small
\begin{align} \label{eq:pdf_delay_Doppler}
	& f\left(\tau ^{n,k}\right) =  \left(1-\gamma\right)\delta\left(\tau ^{n,k}\right) \nonumber \\
	& \qquad \quad + \gamma\frac{1}{\Delta _{\tau}}g_{\left({\bar \tau }^{n,k}-0.5\Delta _{\tau}, {\bar \tau }^{n,k}+0.5\Delta _{\tau}\right)}\left(\tau ^{n,k}\right), \\
	& f\left(f_{d}^{n,k}\right) = \left(1-\gamma\right)\delta\left(f_{d}^{n,k}\right) \nonumber \\
	& \qquad \quad + \gamma\frac{1}{\Delta _{f_{d}}}g_{\left({\bar f }_{d}^{n,k}-0.5\Delta _{f_{d}}, {\bar f }_{d}^{n,k}+0.5\Delta _{f_{d}}\right)}\left(f_{d}^{n,k}\right),
\end{align} \normalsize
 where ${g_{\left( {a,b} \right)}}\left( {{x}} \right) = \left\{ {\begin{array}{*{20}{c}}
		{1,}&{{x} \in \left[ {a,b} \right)}\\
		{0,}&{otherwise}
\end{array}} \right.$ is the window function.

For the Gaussian noise case, the observed delay and Doppler in the $(n,k)$-th cell can be modelled as
\begin{eqnarray}\label{eq:output_delay_Doppler}
	\tilde \tau ^{n,k} = \tau^{n,k} + n _{\tau }^{n,k}, \ \tilde f_{d}^{n,k} = f_{d}^{n,k} + n _{f_{d}}^{n,k},
\end{eqnarray}
where $n _{\tau }^{n,k} \sim \mathcal {N}\left( {0\text{,}  \left(\sigma _{\tau}^{n,k} \right)^2} \right) $ and $n _{f_{d}}^{n,k} \sim \mathcal {N}\left( {0\text{,}  \left(\sigma _{f_d}^{n,k}\right) ^2} \right)$ denote the error of delay and Doppler.
The lower bounds on $\sigma _{\tau}^{n,k} $ and $\sigma _{f_d}^{n,k}$, i.e., Cramér-Rao Bound (CRB), are expressed as\cite{mark2010principles}
\begin{equation}\label{eq:sigma_R_2}
	{\sigma _{\tau}^{n,k}} \ge \frac{1}{{{B_{rms}} \sqrt {{{\rm SINR}^{n,k}}} }}, \ {\sigma _{f_d}^{n,k}} \ge \frac{1 }{{{T_{rms}} \sqrt {{{\rm SINR}^{n,k}}} }},
	%	\vspace*{-5pt} %留空白，可自己调整
\end{equation}
where $T_{rms}$ and $B_{rms}$ are the root mean square (RMS) time duration and the RMS bandwidth of ISAC waveform\cite[Eq. (18.27)]{mark2010principles}. ${\rm SINR}^{n,k}$ is the output SINR of the $(n,k)$-th cell.

Although the range (delay) and velocity (Doppler) parameters of the target are continuous, in most application scenarios, the estimation error only needs to be less than a certain requirement threshold, e.g., the estimation error of the range need to be less than 0.5 m. 
%In other words, it is not required to obtain completely accurate continuous values. 
This means that the interested delay and Doppler of the sensing system can be considered as discrete values as long as the estimation error meets the requirements of the system.

Therefore, the $(n,k)$-th delay-Doppler cell can be divided into multiple sub-cells with the delay interval $D_{\tau }^{n,k}$ and the Doppler interval $D_{f_d }^{n,k}$, analogous to a communication constellation. The target within a delay-Doppler cell must be in one of its sub-cells, as shown in Fig. \ref{fg:radar_efficiency}. Now, the receiver only needs to determine the sub-cell where the target belongs and then uses the corresponding discrete delay/Doppler value of the determined sub-cell as the estimated delay/Doppler of the target, i.e., the delay and Doppler are estimated as	
\begin{align}\label{eq:R_est}
	{{\hat \tau }^{n,k}} &= \left\lfloor {{{\left( {\tilde \tau ^{n,k} - {{\bar \tau }^{n,k}}} \right)} \mathord{\left/
				{\vphantom {{\left( {\tilde \tau ^{n,k} - {{\bar \tau }^{n,k}}} \right)} {{D_{\tau }^{n,k}}}}} \right.
				\kern-\nulldelimiterspace} {{D_{\tau }^{n,k}}}}} \right\rceil {D_{\tau }^{n,k}} + {{\bar \tau }^{n,k}}, \nonumber \\
			{{\hat f}_{d}^{n,k}} &= \left\lfloor {{{\left( {\tilde f_{d}^{n,k} - {{\bar f}_{d}^{n,k}}} \right)} \mathord{\left/
						{\vphantom {{\left( {\tilde f_{d}^{n,k} - {{\bar f}_{d}^{n,k}}} \right)} {{D_{{f_d}}^{n,k}}}}} \right.
						\kern-\nulldelimiterspace} {{D_{{f_d}}^{n,k}}}}} \right\rceil {D_{{f_d}}^{n,k}} + {{\bar f}_{d}^{n,k}}.
\end{align}

\begin{figure}[h] 	
		\vspace*{-5pt} %留空白，可自己调整
	\centerline{\includegraphics[width=2.3 in]{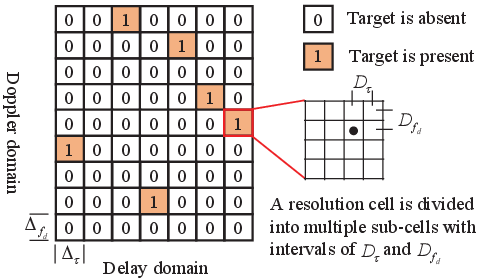}}
	\caption{Delay-Doppler resolution cells and divided sub-cells of sensing system.}
	\label{fg:radar_efficiency}	
%	\vspace*{-1pt} %留空白，可自己调整
\end{figure}

For the given certain requirement thresholds of range and velocity estimation errors, i.e., $\bar D_{R}$ and $\bar D_{V}$, we let $D_{\tau }^{n,k}  = {{2{\bar D_R}} \mathord{\left/
		{\vphantom {{2{\bar D_R}} c}} \right.
		\kern-\nulldelimiterspace} c} = \bar D_{\tau}$ and $D_{f_d }^{n,k}  = {{2{\bar D_V}{f_c}} \mathord{\left/
		{\vphantom {{2{\bar D_V}{f_c}} c}} \right. \kern-\nulldelimiterspace} c}= \bar D_{f_d}$. 
At this time, if the correct sub-cell has been determined, it can ensure that the estimation error of range/velocity is less than the given requirement threshold $\bar D_{R}$/$\bar D_{V}$. Now, the obtained information about the target by estimating delay and Doppler in the $(n,k)$-th cell can be written as following Lemma 1 and Lemma 2.

\begin{lemma}
	If the correct sub-cell has been determined, the obtained information about the target by estimating delay in the $(n,k)$-th cell after detection is given by
	\begin{align}\label{eq:info_est_delay}
		& {I\left(\tau ^{n,k};{\hat \tau }^{n,k}| U^{n,k},{\hat U}^{n,k}\right)} = \gamma P_{D}^{n,k}{\log _2}\left({\frac{{{\Delta _\tau}}}{{\bar D _{\tau } }}} \right).
	\end{align}
\end{lemma}
\begin{proof}
	See Appendix \ref{proof_Lemma1}.
\end{proof}

\begin{lemma} 
	If the correct sub-cell has been determined, the obtained information about the target by estimating Doppler in the $(n,k)$-th cell after detection and delay estimation can be written as
\small
\begin{align}\label{eq:info_est_Doppler}
	& {I\left(f_{d}^{n,k};{\hat f }_{d}^{n,k}| U^{n,k},{\hat U}^{n,k},\tau ^{n,k},{\hat \tau }^{n,k}\right)} = \gamma P_{D}^{n,k}{\log _2}\left({\frac{{{\Delta _{f_d}}}}{{\bar D _{{f_d} } }}} \right).
\end{align}\normalsize
\end{lemma}
\begin{proof}
Since the delay and the Doppler of the target are independent of each other, ${I\left(f_{d}^{n,k};{\hat f }_{d}^{n,k}| U^{n,k},{\hat U}^{n,k},\tau ^{n,k},{\hat \tau }^{n,k}\right)} = {I\left(f_{d}^{n,k};{\hat f }_{d}^{n,k}| U^{n,k},{\hat U}^{n,k}\right)}$. Then, the proof is similar to Lemma 1 by replacing ${\tau ^{n,k}}$ with $f_{d}^{n,k}$.
\end{proof}

{ Lemma 1 and Lemma 2 reveal that the obtained information on the target by estimating delay/Doppler is proportional to the number of divided sub-cells of delay/Doppler, i.e., ${\frac{{{\Delta _\tau}}}{{\bar D _{\tau } }}}$ or ${\frac{{{\Delta _{f_d}}}}{{\bar D _{{f_d} } }}}$. Specifically, the smaller the interval $\bar D_{\tau}$ or $\bar D_{f_d}$ of divided sub-cell, the more information can be obtained by estimation. Moreover, the obtained information is also affected by the probabilities of the target presenting and being correctly detected, i.e., $\gamma$ and $P_{D}^{n,k}$. It is because only when the target is present and correctly detected, the estimation process can obtain the information on the target.}

\begin{figure*}
\vspace*{-10pt} %留空白，可自己调整
	\small
	\begin{flalign}\label{eq:info_sensing}
		&\		I\left( {{U^{n,k}},{\tau ^{n,k}},{f_{d}^{n,k}};{{\hat U}^{n,k}},{{\hat \tau }^{n,k}},{{\hat f}_{d}^{n,k}}} \right) 	= I\left( {{U^{n,k}};{{\hat U}^{n,k}}} \right) + I\left( {{\tau ^{n,k}};{{\hat U}^{n,k}}|{U^{n,k}}} \right) + I\left( {{f_{d}^{n,k}};{{\hat U}^{n,k}}|{U^{n,k}},{\tau ^{n,k}}} \right) & \nonumber\\
		&\	\quad\quad\quad\quad\quad\quad\quad\quad\quad\quad\quad\quad	+ I\left( {{U^{n,k}};{{\hat \tau }^{n,k}}|{{\hat U}^{n,k}}} \right) + I\left( {{\tau ^{n,k}};{{\hat \tau }^{n,k}}|{U^{n,k}},{{\hat U}^{n,k}}} \right) + I\left( {{f_{d}^{n,k}};{{\hat \tau }^{n,k}}|{U^{n,k}},{\tau ^{n,k}},{{\hat U}^{n,k}}} \right) & \nonumber\\
		&\	\quad\quad\quad\quad\quad\quad\quad\quad\quad\quad\quad\quad	+ I\left( {{U^{n,k}};{{\hat f}_{d}^{n,k}}|{{\hat U}^{n,k}},{{\hat \tau }^{n,k}}} \right) + I\left( {{\tau ^{n,k}};{{\hat f}_{d}^{n,k}}|{U^{n,k}},{{\hat U}^{n,k}},{{\hat \tau }^{n,k}}} \right) + I\left( {{f_{d}^{n,k}};{{\hat f}_{d}^{n,k}}|{U^{n,k}},{\tau ^{n,k}},{{\hat U}^{n,k}},{{\hat \tau }^{n,k}}} \right) &  \nonumber \\
		&\ \stackrel{a}=  I\left( {{U^{n,k}};{{\hat U}^{n,k}}} \right) + I\left( {{\tau ^{n,k}};{{\hat \tau }^{n,k}}|{U^{n,k}},{{\hat U}^{n,k}}} \right) + I\left( {{f_{d}^{n,k}};{{\hat f}_{d}^{n,k}}|{U^{n,k}},{\tau ^{n,k}},{{\hat U}^{n,k}},{{\hat \tau }^{n,k}}} \right)  & \nonumber \\
		&\ = {H_b}\left( {\left( {1 - \gamma } \right){P_{fa}} + \gamma {P_{D}^{n,k}}} \right){\rm{ }} - \left( {1 - \gamma } \right){H_b}\left( {{P_{fa}}} \right) - \gamma {H_b}\left( {1 - {P_{D}^{n,k}}} \right) + \gamma {P_{D}^{n,k}}{\log _2} \left( {\frac{{{\Delta _\tau }{\Delta _{{f_d}}}}}{{{\bar D_\tau }{\bar D_{{f_d}}}}}} \right). & 
	\end{flalign}\normalsize
%\hrule
\vspace*{-10pt} %留空白，可自己调整
\end{figure*}
Next, using the chain rule, the obtained information about target by sensing (including detection and estimation) in the $\left(n,k\right)$-th cell, i.e., $I\left( {{U^{n,k}},{\tau ^{n,k}},{f_{d}^{n,k}};{{\hat U}^{n,k}},{{\hat \tau }^{n,k}},{{\hat f}_{d}^{n,k}}} \right)$, is decomposed as (\ref{eq:info_sensing}) on the next page, where the equality of (a) comes from $I\left( {{\tau ^{n,k}};{{\hat U}^{n,k}}|{U^{n,k}}} \right) = 0$, $I\left( {{f_{d}^{n,k}};{{\hat U}^{n,k}}|{U^{n,k}},{\tau ^{n,k}}} \right)=0$, $I\left( {{U^{n,k}};{{\hat \tau }^{n,k}}|{{\hat U}^{n,k}}} \right)=0$, $I\left( {{f_{d}^{n,k}};{{\hat \tau }^{n,k}}|{U^{n,k}},{\tau ^{n,k}},{{\hat U}^{n,k}}} \right)=0$, $I\left( {{U^{n,k}};{{\hat f}_{d}^{n,k}}|{{\hat U}^{n,k}},{{\hat \tau }^{n,k}}} \right)=0$, and $I\left( {{\tau ^{n,k}};{{\hat f}_{d}^{n,k}}|{U^{n,k}},{{\hat U}^{n,k}},{{\hat \tau }^{n,k}}} \right)=0$.

\begin{definition}
	SSE is defined as how much interested information on targets can be actually obtained by sensing (including detection and estimation) for a given waveform with the time interval $T_A$ and bandwidth $B$, which is given by 
	\small
		\begin{equation}
			{\eta _{sen}} = \frac{1}{{B{T_A}}}\sum\limits_{n = 0}^{{N_\tau } - 1} {\sum\limits_{k = 0}^{{N_{{f_d}}} - 1} {I\left( {{U^{n,k}},{\tau ^{n,k}},{f_{d}^{n,k}};{{\hat U}^{n,k}},{{\hat \tau }^{n,k}},{{\hat f}_{d}^{n,k}}} \right)}}.
		\end{equation}
	\normalsize
\end{definition}

According to (\ref{eq:info_sensing}), SSE can be rewritten as (\ref{eq:eta_sen}).  
\begin{figure*}[htbp]
			\vspace*{-10pt} %留空白，可自己调整
	\small
	\begin{flalign}\label{eq:eta_sen}
&\	{\eta _{sen}}  = \frac{1}{{B{T_A}}}\sum\limits_{n = 0}^{{N_{\tau} } - 1} {\sum\limits_{k = 0}^{{N_{{f_d}}} - 1} {{{H_b}\left( {\left( {1 - \gamma } \right){P_{fa}} + \gamma {P_{D}^{n,k}}} \right){\rm{ }} - \left( {1 - \gamma } \right){H_b}\left( {{P_{fa}}} \right) - \gamma {H_b}\left( {1 - {P_{D}^{n,k}}} \right) + \gamma {P_{D}^{n,k}}{\log _2} \left( {\frac{{{\Delta _\tau }{\Delta _{{f_d}}}}}{{{\bar D_\tau }{\bar D_{{f_d}}}}}} \right)}}}.&
	\end{flalign}\normalsize
%\hrule
	\vspace*{-10pt} %留空白，可自己调整
\end{figure*} 
When the probability of detection ${P_{D}^{n,k}}$ of each present target is equal to the pre-set threshold (lower bound) of sensing system, i.e., $P_D^{\rm Th}$, a lower bound on SSE can be obtained shown as (\ref{eq:eta_sen_2}). 
\begin{figure*}[htbp]
			\vspace*{-10pt} %留空白，可自己调整	
	\small
		\begin{flalign}\label{eq:eta_sen_2}		
			&\			{\eta _{sen}} = \frac{1}{{B{T_A}}}\frac{{\tau _{\max }}}{{{\Delta _ \tau}}}\frac{{2{f_{d,\max }}}}{{{\Delta _{f_d}}}}\left[{H\left( {\left( {1 - \gamma } \right){P_{fa}} + \gamma {P_D^{\rm Th}}} \right) - \left( {1 - \gamma } \right)H\left( {{P_{fa}}} \right) - \gamma H\left( {1 - {P_D^{\rm Th}}} \right) + \gamma{P_D^{\rm Th}}{\log _2}\left({\frac{{{\Delta _\tau}{\Delta _{f_d} }}}{{{\bar D _{\tau}}{\bar D _{f_d}}}}} \right)}\right].&
	\end{flalign}	\normalsize
%\hrule
	\vspace*{-10pt} %留空白，可自己调整
\end{figure*}

\begin{example} \label{ex:SSE}
	When $\lambda = 0.5$, $P_{fa} = 10^{-6}$ and $P_D^{\rm Th} = 0.999$, the lower bound on SSE in (\ref{eq:eta_sen_2}) can be rewritten as
			%		\vspace*{-10pt} %留空白，可自己调整	
		\begin{align}\label{eq:eta_sen_3}		
			{\eta _{sen}} &= \frac{1}{{B{T_A}}}\frac{{\tau _{\max }}}{{{\Delta _ \tau}}}\frac{{2{f_{d,\max }}}}{{{\Delta _{f_d}}}}\left[0.9943 + 0.4995{\log _2}\left({\frac{{{\Delta _\tau}{\Delta _{f_d} }}}{{{\bar D _{\tau}}{\bar D _{f_d}}}}} \right)\right].
		\end{align}	
%		\vspace*{-20pt} %留空白，可自己调整	
\end{example}

Next, we analyze the similarity and difference between our proposed SSE and the existing metrics in terms of the radar capacity\cite{guerci2015joint}, the radar estimation rate\cite{chiriyath2015inner} and the theoretical sensing rate\cite{choi2023information}.

{\bf{Similarity}}: They measure how much uncertainty can be eliminated by the once-sensing process from the information theory perspective.

{\bf{Difference}}: \romannumeral 1) The radar capacity is determined by the maximum number of resolvable targets, which does not take into account the additional target information that can be obtained within each resolution unit. \romannumeral 2) The radar estimation rate measures how much uncertainty is cancelled by the parameter estimation process, which is affected by not only the SINR, but also the predicted tracking error. It means that enlarging the predicted tracking error can improve the radar estimation rate. \romannumeral 3) The theoretical sensing rate measures the fundamental limit on the information obtained by a pulse-Doppler radar system, which does not consider the actual ISAC waveform. \romannumeral 4) Our proposed SSE evaluates an achievable resource utilization efficiency of the sensing for the given ISAC waveform, considering both detection and estimation processes.

\vspace{1ex}
\noindent\emph{(2) Proposed SOP}
\vspace{1ex}

{ We next define SOP to measure the reliability of sensing functionality.}
Due to the presence of noise, an incorrect sub-cell may be determined, which results in the estimation error of range/velocity larger than the given requirement threshold. Similar to the BER measuring reliability of communications, we define SOP to evaluate the reliability of sensing results.

\begin{definition}
	We define that a sensing outage event occurs when the sensing result cannot meet the requirement of system, including false alarm, missed alarm, and the estimated error of delay/Doppler larger than the given requirement threshold. Consequently, SOP is defined as the number of sensing outage divided by the total number of sensing, which is expressed as (\ref{eq:sen_error}).
	\begin{figure*}[htbp]
				\vspace*{-10pt} %留空白，可自己调整	
		\small
			\begin{flalign}\label{eq:sen_error}		
				&\		{P_{e,sen}} = \frac{1}{{{N_\tau }{N_{{f_d}}}}}\sum\limits_{n = 0}^{{N_\tau } - 1} {\sum\limits_{k = 0}^{{N_{{f_d}}} - 1} {\left[ {\left( {1 - \gamma } \right){P_{fa}} + \gamma \left( {1 - {P_{D}^{n,k}}} \right) + \gamma {P_{D}^{n,k}} \Pr \left( {\left| {\tau ^{n,k} - {{\hat \tau }^{n,k}}} \right| \ge {\bar D_{\tau}}\parallel \left| {{f}_{d}^{n,k} - {{\hat f}_{d}^{n,k}}} \right| \ge {\bar D_{f_{d}}}} \right)} \right]} }. &		
		\end{flalign}	\normalsize
	\hrule
			\vspace*{-10pt} %留空白，可自己调整
	\end{figure*} 

\end{definition}

In (\ref{eq:sen_error}), $\Pr \left( {\left| {\tau ^{n,k} - {{\hat \tau }^{n,k}}} \right| \ge {\bar D_{\tau}}\parallel \left| {{f}_{d}^{n,k} - {{\hat f}_{d}^{n,k}}} \right| \ge {\bar D_{f_{d}}}} \right)$ denotes the probability of the estimated error of delay/Doppler greater than the given requirement thresholds when ${U^{n,k}}=1$ and ${{\hat U}^{n,k}}=1$, which can be rewritten as
\small
\begin{align}\label{eq:est_error}
	&\Pr \left( {\left| {\tau ^{n,k} - {{\hat \tau }^{n,k}}} \right| \ge {\bar D_{\tau}}\parallel \left| {{f}_{d}^{n,k} - {{\hat f}_{d}^{n,k}}} \right| \ge {\bar D_{f_{d}}}} \right) \nonumber \\
	&= 1 - \Pr \left( {\tau ^{L} < {{\tilde \tau }^{n,k}} < \tau ^{U}} \right)\Pr \left( f_{d}^L < {{\tilde f}_{d}^{n,k}} < f_{d}^U \right),
\end{align}	\normalsize
where $\tau ^{L} = \left\lfloor {{{\left( {{\tau ^{n,k}} - {{\bar \tau }^{n,k}}} \right)} \mathord{\left/
				{\vphantom {{\left( {{\tau ^{n,k}} - {{\bar \tau }^{n,k}}} \right)} {{\bar D_\tau }}}} \right.
				\kern-\nulldelimiterspace} {{\bar D_\tau }}}} \right\rfloor {\bar D_\tau } + {{\bar \tau }^{n,k}} - 0.5{\bar D_\tau }$, $\tau ^{U} = \left\lceil {{{\left( {{\tau ^{n,k}} - {{\bar \tau }^{n,k}}} \right)} \mathord{\left/
				{\vphantom {{\left( {{\tau ^{n,k}} - {{\bar \tau }^{n,k}}} \right)} {{\bar D_\tau }}}} \right.
				\kern-\nulldelimiterspace} {{\bar D_\tau }}}} \right\rceil {\bar D_\tau } + {{\bar \tau }^{n,k}} + 0.5{\bar D_\tau }$, $f_{d}^{L} = \left\lfloor {{{\left( {{f_{d}^{n,k}} - {{\bar f}_{d}^{n,k}}} \right)} \mathord{\left/
				{\vphantom {{\left( {{f_{d}^{n,k}} - {{\bar f}_{d}^{n,k}}} \right)} {{\bar D_{{f_d}}}}}} \right.
				\kern-\nulldelimiterspace} {{\bar D_{{f_d}}}}}} \right\rfloor {\bar D_{{f_d}}} + {{\bar f}_{d}^{n,k}} - 0.5{\bar D_{{f_d}}}$, and $f_{d}^{U} = \left\lceil {{{\left( {{f_{d}^{n,k}} - {{\bar f}_{d}^{n,k}}} \right)} \mathord{\left/
				{\vphantom {{\left( {{f_{d}^{n,k}} - {{\bar f}_{d}^{n,k}}} \right)} {{\bar D_{{f_d}}}}}} \right.
				\kern-\nulldelimiterspace} {{\bar D_{{f_d}}}}}} \right\rceil {\bar D_{{f_d}}} + {{\bar f}_{d}^{n,k}} + 0.5{\bar D_{{f_d}}}$. 
According to (\ref{eq:pdf_delay_Doppler}) and (\ref{eq:output_delay_Doppler}), when $U^{n,k}=1$, the pdfs of ${\tilde \tau }^{n,k}$ and ${\tilde f}_{d}^{n,k}$ can be written as
 \small
\begin{align}	
&{f}\left( {{{\tilde \tau }^{n,k}}|U^{n,k}=1} \right) = \int_{ - \infty }^\infty  \hspace{-5px}{{f}\left( {{\tau ^{n,k}}}|U^{n,k}=1 \right){f}\left( {{{\tilde \tau }^{n,k}} - {\tau ^{n,k}}} \right)d{\tau ^{n,k}}}  \nonumber \\
&= \int_{{{\bar \tau }^{n,k}} - 0.5{\Delta _\tau }}^{{{\bar \tau }^{n,k}} + 0.5{\Delta _\tau }} {\frac{1}{{{\Delta _\tau }}}\frac{1}{{\sqrt {2\pi } {\sigma _{\tau }^{n,k}}}}{e^{ - \left(\frac{ {{{\tilde \tau }^{n,k}} - {\tau }} }{{\sigma _{\tau }^{n,k}}}  \right)^2}}d{\tau }},
\end{align} \normalsize
and
 \small
	\begin{align}	
{f}\left( {{{\tilde f}_{d}^{n,k}}}|U^{n,k}=1 \right) = \int_{{{\bar f}_{d}^{n,k}} - 0.5{\Delta _{{f_d}}}}^{{{\bar f}_{d}^{n,k}} + 0.5{\Delta _{{f_d}}}} {\frac{1}{{{\Delta _{{f_d}}}}}\frac{1}{{\sqrt {2\pi } {\sigma _{{f_d}}^{n,k}}}}{e^{ - \left( \frac{{{ {{{\tilde f}_{d}^{n,k}} - {f_{d}}} }}}{{\sigma _{{f_d}}^{n,k}}}  \right)^2 } }d{f_{d}}} .
\end{align} \normalsize

As a result, SOP in (\ref{eq:sen_error}) can be written as (\ref{eq:sen_error_2}), where $P_{\det}=\left( {1 - \gamma } \right){P_{fa}} + \gamma \left( {1 - {P_{D}^{n,k}}} \right)$, and $Q\left( x \right) = \int_x^\infty  {\frac{1}{{\sqrt {2\pi } }}{e^{ - \frac{{{t^2}}}{2}}}dt}$ denoting Q-function.
\begin{figure*}[htbp]
			\vspace*{-10pt} %留空白，可自己调整	
	\small
		\begin{flalign}\label{eq:sen_error_2}		
&\				{P_{e,sen}} = \frac{1}{{{N_\tau }{N_{{f_d}}}}}\sum\limits_{n = 0}^{{N_\tau } - 1} \sum\limits_{k = 0}^{{N_{{f_d}}} - 1} \left\{ P_{\det} + \gamma {P_{D}^{n,k}} \left[ 1 - \frac{1}{{{\Delta _\tau }{\Delta _{{f_d}}}}}\int_{{{\bar \tau }^{n,k}} - 0.5{\Delta _\tau }}^{{{\bar \tau }^{n,k}} + 0.5{\Delta _\tau }} {\left( {Q\left( {\frac{{\tau ^{L} - {\tau ^{n,k}}}}{{{\sigma _{\tau }^{n,k}}}}} \right) - Q\left( {\frac{{\tau ^{U} - {\tau ^{n,k}}}}{{{\sigma _{\tau }^{n,k}}}}} \right)} \right)d{\tau ^{n,k}}} \right.\ \right.\ & \nonumber \\
&\			\quad\quad\quad\quad\quad\quad\quad\quad\quad\quad\quad\quad\quad\quad\quad\quad\quad\quad\quad\quad\quad\quad\quad	  \left.\ \left.\ \times \int_{{{\bar f}_{d}^{n,k}} - 0.5{\Delta _{{f_d}}}}^{{{\bar f}_{d}^{n,k}} + 0.5{\Delta _{{f_d}}}} {\left( {Q\left( {\frac{{f_{d}^D - {f_{d}^{n,k}}}}{{{\sigma _{{f_d}}^{n,k}}}}} \right) - Q\left( {\frac{{f_{d}^U - {f_{d}^{n,k}}}}{{{\sigma _{{f_d}}^{n,k}}}}} \right)} \right)d{f_{d}^{n,k}}}  \right] \right\}  .&
		\end{flalign} \normalsize	
%	\hrule
	\vspace*{-10pt} %留空白，可自己调整
\end{figure*} 

\begin{corollary}
	SOP, i.e., ${P_{e,sen}}$, is lower bounded by ${P_{e,sen}^{\min}}$ and upper bounded by ${P_{e,sen}^{\max}}$, i.e.,
	\begin{eqnarray}
		P_{e,sen}^{\min }\mathop  \le \limits^{\left( a \right)} {P_{e,sen}}\mathop  \le \limits^{\left( b \right)} P_{e,sen}^{\max },
	\end{eqnarray}
	where $P_{e,sen}^{\min}$ and $P_{e,sen}^{\max}$ are shown in (\ref{eq:SOP_min}) and (\ref{eq:SOP_max}), respectively. The equality of (a) holds if and only if ${\tau ^{n,k}} = \left\lceil {{{\left( {{\tau ^{n,k}} - {{\bar \tau }^{n,k}}} \right)} \mathord{\left/
				{\vphantom {{\left( {{\tau ^{n,k}} - {{\bar \tau }^{n,k}}} \right)} {{\bar D_\tau }}}} \right.
				\kern-\nulldelimiterspace} {{\bar D_\tau }}}} \right\rceil {\bar D_\tau } + {{\bar \tau }^{n,k}} - 0.5{\bar D_\tau }$ and ${f_{d}^{n,k}} = \left\lceil {{{\left( {{f_{d}^{n,k}} - {{\bar f}_{d}^{n,k}}} \right)} \mathord{\left/
				{\vphantom {{\left( {{f_{d}^{n,k}} - {{\bar f}_{d}^{n,k}}} \right)} {{\bar D_{{f_d}}}}}} \right.
				\kern-\nulldelimiterspace} {{\bar D_{{f_d}}}}}} \right\rceil {\bar D_{{f_d}}} + {{\bar f}_{d}^{n,k}} - 0.5{\bar D_{{f_d}}}$. The equality of (b) holds if and only if ${\tau ^{n,k}} = \left\lfloor {{{\left( {{\tau ^{n,k}} - {{\bar \tau }^{n,k}}} \right)} \mathord{\left/
				{\vphantom {{\left( {{\tau ^{n,k}} - {{\bar \tau }^{n,k}}} \right)} {{D_\tau }}}} \right.
				\kern-\nulldelimiterspace} {{\bar D_\tau }}}} \right\rfloor {\bar D_\tau } + {{\bar \tau }^{n,k}}$ and ${f_{d}^{n,k}} = \left\lfloor {{{\left( {{f_{d}^{n,k}} - {{\bar f}_{d}^{n,k}}} \right)} \mathord{\left/
				{\vphantom {{\left( {{f_{d}^{n,k}} - {{\bar f}_{d}^{n,k}}} \right)} {{\bar D_{{f_d}}}}}} \right.
				\kern-\nulldelimiterspace} {{\bar D_{{f_d}}}}}} \right\rfloor {\bar D_{{f_d}}} + {{\bar f}_{d}^{n,k}}$.
\end{corollary}
\begin{proof}
	See Appendix \ref{proof_Theorem1}.
\end{proof}
\begin{figure*}[htbp]
	\vspace*{-10pt} %留空白，可自己调整
	\small
\begin{flalign}\label{eq:SOP_min}
&\	P_{e,sen}^{\min } = \frac{1}{{{N_\tau }{N_{{f_d}}}}}\sum\limits_{n = 0}^{{N_\tau } - 1} {\sum\limits_{k = 0}^{{N_{{f_d}}} - 1} {\left\{ {P_{\det} + \gamma P_{D}^{n,k}\left[ {2Q\left( {\frac{{{\bar D_\tau }}}{{{\sigma _{\tau }^{n,k}}}}} \right) + 2Q\left( {\frac{{{\bar D_{{f_d}}}}}{{{\sigma _{{f_d}}^{n,k}}}}} \right) - 4Q\left( {\frac{{{\bar D_\tau }}}{{{\sigma _{\tau }^{n,k}}}}} \right)Q\left( {\frac{{{\bar D_{{f_d}}}}}{{{\sigma _{{f_d}}^{n,k}}}}} \right)} \right]} \right\}} }.&
\end{flalign}
\vspace*{-10pt} %留空白，可自己调整
\begin{flalign}\label{eq:SOP_max}
&\		P_{e,sen}^{\max } = \frac{1}{{{N_\tau }{N_{{f_d}}}}}\sum\limits_{n = 0}^{{N_\tau } - 1} {\sum\limits_{k = 0}^{{N_{{f_d}}} - 1} {\left\{ {P_{det} + \gamma {P_{D}^{n,k}} \left[ {2Q\left( {\frac{{{\bar D_\tau }}}{{2{\sigma _{\tau }^{n,k}}}}} \right) + 2Q\left( {\frac{{{\bar D_{{f_d}}}}}{{2{\sigma _{{f_d}}^{n,k}}}}} \right) - 4Q\left( {\frac{{{\bar D_\tau }}}{{2{\sigma _{\tau }^{n,k}}}}} \right)Q\left( {\frac{{{\bar D_{{f_d}}}}}{{2{\sigma _{{f_d}}^{n,k}}}}} \right)} \right]} \right\}} }. &
\end{flalign}\normalsize
\hrule
	\vspace*{-10pt} %留空白，可自己调整
\end{figure*} 

{ Corollary 1 reveals that increasing the values of ${{{\bar D_\tau }} \mathord{\left/
			{\vphantom {{{\bar D_\tau }} {{\sigma _{\tau }^{n,k}}}}} \right.
			\kern-\nulldelimiterspace} {{\sigma _{\tau }^{n,k}}}}$ and ${{{\bar D_{{f_d}}}} \mathord{\left/
			{\vphantom {{{\bar D_{{f_d}}}} {{\sigma _{{f_d}}^{n,k}}}}} \right.
			\kern-\nulldelimiterspace} {{\sigma _{{f_d}}^{n,k}}}}$ can decrease both lower bound and upper bound of ${P_{e,sen}}$ according to the proporty of the Q-function. In other words, enlarging the intervals of delay and Doppler, i.e., ${\bar D_\tau }$ and $\bar D_{{f_d}}$ can decrease the SOP, however, which loses the SSE according to (\ref{eq:eta_sen_2}). Moreover, reducing the estimation error of delay and Doppler, i.e., ${\sigma _{\tau }^{n,k}}$ and ${\sigma _{{f_d}}^{n,k}}$ can also decrease the SOP, with the cost of increasing SINR requirement.} 

\begin{example}
		Following Example \ref{ex:SSE}, if ${{{\bar D_\tau }} \mathord{\left/
				{\vphantom {{{\bar D_\tau }} {{\sigma _{\tau }^{n,k}}}}} \right.
				\kern-\nulldelimiterspace} {{\sigma _{\tau }^{n,k}}}} = {{{\bar D_{{f_d}}}} \mathord{\left/
				{\vphantom {{{\bar D_{{f_d}}}} {{\sigma _{{f_d}}^{n,k}}}}} \right.
				\kern-\nulldelimiterspace} {{\sigma _{{f_d}}^{n,k}}}}=4$, it can get $P_{e,sen}^{max}=6.67 \times 10^{-2}$ and $P_{e,sen}^{min}=5.59 \times 10^{-4}$. If ${{{\bar D_\tau }} \mathord{\left/
				{\vphantom {{{\bar D_\tau }} {{\sigma _{\tau }^{n,k}}}}} \right.
				\kern-\nulldelimiterspace} {{\sigma _{\tau }^{n,k}}}} = {{{\bar D_{{f_d}}}} \mathord{\left/
				{\vphantom {{{\bar D_{{f_d}}}} {{\sigma _{{f_d}}^{n,k}}}}} \right.
				\kern-\nulldelimiterspace} {{\sigma _{{f_d}}^{n,k}}}}=6$, it can get $P_{e,sen}^{max}=4.50 \times 10^{-3}$ and $P_{e,sen}^{min}=5.00 \times 10^{-4}$.  It means that for fixed ${\sigma _{\tau }^{n,k}}$ and ${\sigma _{{f_d}}^{n,k}}$, both $P_{e,sen}^{max}$ and $P_{e,sen}^{min}$ decrease with the increase of ${\bar D_\tau }$ and/or ${\bar D_{{f_d}}}$.
\end{example}

{\bf{Discussion about the proposed SSE and SOP:}} 

(\romannumeral1) The traditional sensing performances (e.g., resolution, unambiguous delay/Doppler, probability of detection, probability of false alarm, and estimation error) can be reflected in SSE and SOP, according to (\ref{eq:eta_sen_2}), (\ref{eq:SOP_min}) and (\ref{eq:SOP_max}). Specifically, for the given $B$ and $T_A$, the delay and Doppler resolutions of ISAC waveform keep constant. At this time, SSE is proportional to the probability of detection $P_{D}$ and the unambiguous delay/Doppler $\tau _{max} $/$f_{d,max}$, but SSE is inversely proportional to the probability of false alarm $P_{fa}$. On the other hand, SOP is proportional to the probability of false alarm $P_{fa}$ and inversely proportional to the probability of detection $P_{D}$. Moreover, SOP is proportional to the CRB on estimation error of delay/Doppler, i.e., ${\sigma _{\tau }^{n,k}}$/${\sigma _{{f_d}}^{n,k}}$.

(\romannumeral2) The trade-off between SSE and SOP is analogous to the trade-off between CSE and BER. Specifically, decreasing $\bar D_{\tau}$ and $\bar D_{f_d}$, i.e., dividing one resolution cell into more sub-cells, can improve SSE according to (\ref{eq:eta_sen_2}). However, $k_{\tau} = {{\bar D_\tau }  \mathord{\left/
		{\vphantom {{\bar D_\tau }  { {{\sigma _{\tau }^{n,k}}} }}} \right.
		\kern-\nulldelimiterspace} { {{\sigma _{\tau }^{n,k}}} }}$ and $k_{f_d} =  {{\bar D_{{f_d}}}  \mathord{\left/
		{\vphantom {{\bar D_{{f_d}}}  { {{\sigma _{{f_d}}^{n,k}}} }}} \right.
		\kern-\nulldelimiterspace} { {{\sigma _{{f_d}}^{n,k}}} }}$ will decrease for fixed ${\sigma _{\tau }^{n,k}}$ and ${\sigma _{{f_d}}^{n,k}}$, which makes SOP increase. If we want to keep SOP constant, it needs to reduce ${\sigma _{\tau }^{n,k}}$ and ${\sigma _{{f_d}}^{n,k}}$, which requires a higher SINR as a cost. Similarly, increasing the modulation order can improve CSE for communications, but a higher SINR is also required to keep BER constant.

(\romannumeral3) SSE and SOP have similar structures and physical meanings to the existing CSE and BER of communications, respectively. Specifically, the dimensions of SSE and CSE are bit/(s$\cdot$Hz). Both SOP and BER indicate the probability that the sensing/communications result cannot meet the requirements of sensing/communication functionality.

The above properties may facilitate analyzing the performance trade-off of ISAC system and optimizing ISAC waveforms. It motivates us to propose a metric set using the proposed SSE and SOP as well as the existing metrics to simultaneously measure the efficiency, the operating range and the reliability of ISAC system.

\vspace{1ex}
\noindent\emph{(3) Proposed metric set for ISAC system}
\vspace{1ex}

\textbf{Metric for efficiency:} We define the weighted CSE and SSE as the metric to measure the resource utilization efficiency of ISAC system, which is given by
\begin{eqnarray}
\eta _{sc}  = {k _{\eta}}{\eta _{com}} + \left( {1 - {k _{\eta}}} \right){\eta _{sen}},
\end{eqnarray} 
where ${k _{\eta}}$ is the weighted factor of efficiency.

\textbf{Metric for operating range:} The delay/Doppler resolution measures the ability to resolve the smallest delay/Doppler difference of communications  channels and sensing targets, representing the lower bound on operating range of ISAC. Moreover, the unambiguous range/velocity of sensing and the maximum tolerable delay/Doppler of communications belong to the metrics measuring the upper bound on operating range of ISAC. The range/velocity can be described by delay/Doppler. Hence, we let the maximum tolerable delay $\tau _{max}$, the maximum tolerable Doppler $f_{d,max}$, the delay resolution $\Delta _{\tau}$, and the Doppler resolution $\Delta _{f_d}$ be the metrics to measure the operating range of ISAC system. The values of $\eta _{sc}$, $\tau _{max}$, $f_{d,max}$, $\Delta _{\tau}$, and $\Delta _{f_d}$ are affected by parameters of ISAC waveform.

{ We briefly discuss the difference between the ambiguity function and our adopted maximum tolerable delay/Doppler as well as delay/Doppler resolution. The relations between the ambiguity function and other sensing and communications metrics are not intuitive, making it difficult to reveal the analytical trade-offs based on the ambiguity function. On the contrary, there are more intuitive relations among our adopted maximum tolerable delay/Doppler as well as delay/Doppler resolution and other sensing metrics and communications metrics. Therefore, we use these metrics to measure the ranges of the ISAC system and derive the analytical trade-offs.}

\textbf{Metric for reliability:} We define the weighted BER and SOP as the metric to measure the reliability of ISAC system, which is written as
\begin{eqnarray}
	P _{e,sc}  = {k _{p}}{P _{e,com}} + \left( {1 - {k _{p}}} \right){P _{e,sen}},
\end{eqnarray} 
where ${k _{p}}$ is the weighted factor of reliability, and ${P _{e,com}}$ denotes BER of communications. The value of $P _{e,sc}$ is affected by not only parameters of ISAC waveform but also SINR.

\section{Performance Analyses for AFDM-ISAC System} \label{sec_design}

{ In this section, the performances of the AFDM-ISAC system are analyzed. We first derive the analytical relationship between metrics and AFDM waveform parameters for the AFDM-ISAC system, which is the cornerstone of performance analyses. Based on this, the performance trade-offs of the AFDM-ISAC system are analyzed.}

\subsection{Analytical Relationship Between Metrics and AFDM Waveform Parameters}

{ This subsection shows the analytical relationship between AFDM waveform parameters and metrics regarding the maximum tolerable delay/Doppler, SSE and CSE.}

For the maximum tolerable delay, to ensure that the cyclic shift values caused by different delays of paths/targets are different, the maximum normalized delay $l_{max}$ needs to satisfy the following constraint\cite{bemani2023affine}
\begin{equation}
	2N{c_1}\left( {{l_{\max }} + 1} \right) \le N.
\end{equation}
Its physical meaning is that the cyclic shift value caused by the maximum tolerable delay will not exceed the cycle period $N$. In addition, to eliminate inter-symbol interference caused by multi-path delay, it is necessary to ensure that the maximum normalized delay does not exceed the length of CPP, i.e., ${l_{\max }} \le N_{cp}$. 
Hence, the maximum tolerable delay ${\tau _{max}}$ of the AFDM-ISAC system can be expressed as
\begin{align}\label{eq:t_m_AFDM}
	{\tau _{max}} = \frac{1}{{B}}{l_{max}} = \left\{ {\begin{array}{*{20}{c}}
			{\frac{1}{{B}}\left( {\frac{1}{{2{c_1}}} - 1} \right),}&\hspace*{-7px}{{c_1} > \frac{1}{{2\left( {{N_{cp}}{\rm{ + }}1} \right)}},}\\
			{\frac{{{N_{cp}}}}{{B}},}&\hspace*{-7px}{{c_1} \le \frac{1}{{2\left( {{N_{cp}}{\rm{ + }}1} \right)}}.}
	\end{array}} \right. 
\end{align}

For the maximum tolerable Doppler, to ensure that the cyclic shift values caused by Dopplers of two adjacent resolvable paths/targets do not overlap, the cyclic shift values caused by Doppler must be smaller than $2Nc_1$, i.e., the maximum normalized Doppler ${\alpha _{\max }}$ is bounded by
\begin{equation}
	2\left( {{\alpha _{\max }} + {\xi _{v}}} \right) \le 2N{c_1} - 1,
\end{equation}
where ${\xi _v}$ denotes the guard interval setting for fractional normalized Doppler\cite{bemani2023affine}.
Hence, the maximun tolerable Doppler $f_{d,max}$  of the AFDM-ISAC system can be given by
\begin{equation}\label{eq:fd_m_AFDM}
	f_{d,max} ={\alpha _{\max }}\Delta_f = \frac{B}{N}\left( {N{c_1} - {\xi _v} - \frac{1}{2}} \right).
\end{equation}

For SSE, given time interval $T_A$ and bandwidth $B$, the resolutions of range and velocity of the AFDM-ISAC system are $\Delta_{\tau} = {1 \mathord{\left/
		{\vphantom {1 {B}}} \right.
		\kern-\nulldelimiterspace} {B}}$ and $\Delta_{f_d} = {1  \mathord{\left/
		{\vphantom {1  { {{T_A}} }}} \right.
		\kern-\nulldelimiterspace} { {{T_A}} }}$. According to Eq. (\ref{eq:eta_sen_2}), SSE of AFDM waveform is expressed as
\begin{align}\label{eq:eta_r_AFDM}
	&{\eta _{sen}} = 2{\tau _{max}}{f_{d,max}}I_{sen}   \nonumber \\
	&=  \left\{ {\begin{array}{*{20}{c}}
			{\frac{{1 - 2{c_1}}}{{{c_1}}}\left( {{c_1} - \frac{{2{\xi _v} + 1}}{{2N}}} \right)I_{sen} ,}&{{c_1} > \frac{1}{{2\left( {{N_{cp}} + 1} \right)}},}\\
			{2{N_{cp}}\left( {{c_1} - \frac{{2{\xi _v} + 1}}{{2N}}} \right)I_{sen} ,}&{{c_1} \le \frac{1}{{2\left( {{N_{cp}} + 1} \right)}}.}
	\end{array}} \right.
\end{align}
where $I_{sen} = H_b\left( {\left( {1 - \gamma } \right){P_{fa}} + \gamma {P_D^{\rm Th}}} \right) - \left( {1 - \gamma } \right)H_b\left( {{P_{fa}}} \right) - \gamma H_b\left( {1 - {P_D^{\rm Th}}} \right) + \gamma{P_D^{\rm Th}}{\log _2}\left({\frac{1}{BT_A{{\bar D _{\tau}}{\bar D _{f_d}}}}} \right)$.

For CSE, given time interval $T_A$ and bandwidth $B$, the number of transmitted AFDM symbols satisfies ${T_A} = {{{N_{sym}}\left( {N + {N_{cp}}} \right)} \mathord{\left/
		{\vphantom {{{N_{sym}}N_s} B}} \right.
		\kern-\nulldelimiterspace} B}$. In each AFDM symbol, the minimum number of pilot symbols and guard symbols $N_{pg,min}=1$ for the guard intervals free case\cite{zhou2024GI}, and thus the maximum number of valued data symbols $M=N-1$. Substituting it into (\ref{eq:CSE}), CSE of AFDM waveform is given by 
	\begin{align}\label{eq:eta_c_AFDM}
		{\eta _{com}} = \frac{\left(N-1\right)}{{BT_A}} {N_{sym}}{\log _2}{M_{mod }}  = \frac{\left(N-1\right){\log _2}{M_{mod }} }{{N + {N_{cp}}}}.
	\end{align}
At this time, the weighted CSE and SSE is written as Eq. (\ref{eq:sum_effi_2}).

\subsection{Performace Trade-offs Analyses of AFDM-ISAC System} \label{wave_design_ana}

\vspace{1ex}
\noindent\emph{(1) Trade-offs between sensing performaces}
\vspace{1ex}

{ Firstly, we analyze the trade-off of sensing performances between (\romannumeral1) SSE and SOP, (\romannumeral2) ${\tau _{max}}$ and $f_{d,max}$.}
According to (\ref{eq:SOP_max}) and (\ref{eq:eta_r_AFDM}), there is a trade-off between our proposed SSE and SOP for sensing functionality. For the given $B$, $T_A$, $P_{fa}$ and ${\rm SINR^{n,k}}$ in the $\left(n,k\right)$-th delay-Doppler cell, the $P_D$, ${\sigma _{\tau }^{n,k}}$ and ${\sigma _{{f_d}}^{n,k}}$ keep constant. At this time, SSE is traded off against SOP by changing parameter $\bar D_{\tau}$ and/or $\bar D_{f_d}$. 

\begin{example}\label{ex:SSE_vs_SOP}
	Waveform parameters are $M_{mod}=16$, ${N_{cp}}=574$, $N=8192$, $N_{sym} = 28$, $\xi _v=4$, $k_{\tau}=k_{f_d} \in \left[1,2,\cdots, 8\right]$, SINR is equal to $10, 15, 20, 25$ dB, and other parameters are set as Example \ref{ex:SSE}. The trade-offs between SSE and the upper bound on SOP of AFDM-ISAC and OFDM-ISAC systems are shown in Fig. \ref{fg:Pe_sen_vs_effi_sen}, where $c_1 = {{{\left(2\xi _v+1\right)}} \mathord{\left/
			{\vphantom {{{\left(2\xi _v+1\right)}} \left(2N\right)}} \right.
			\kern-\nulldelimiterspace} \left(2N\right)} + {{{3}} \mathord{\left/
			{\vphantom {{{3}} \left(2N\right)}} \right.
			\kern-\nulldelimiterspace} \left(2N\right)}$, $c_2=0$ for AFDM-ISAC system, and $c_1=c_2=0$ for OFDM-ISAC system. { Each curve depicts the trade-off between SSE and the upper bound on SOP at a given SINR. For a given SINR, ${\sigma _{\tau }^{n,k}}$ and ${\sigma _{{f_d}}^{n,k}}$ are calculated by (\ref{eq:sigma_R_2}). At this time, increasing $k_{\tau}$ and $k_{f_d}$ means the growth of $\bar D_{\tau}$ and $\bar D_{f_d}$, which improves the SOP performance and decreases SSE according to (\ref{eq:SOP_max}) and (\ref{eq:eta_r_AFDM}). For example, for AFDM with a 25 dB of SINR, when $k_{\tau}$ and $k_{f_d}$ increase from 1 to 8, SOP performance is improved from $4\times10^{-1}$ and $6\times10^{-4}$, and SSE decreases from 1.2 bit/(s$\cdot$Hz) to 0.56 bit/(s$\cdot$Hz), as shown in Fig. \ref{fg:Pe_sen_vs_effi_sen}.} Therefore, we can observe that SOP performance improve with the decrease of SSE or the increase of SINR for both AFDM-ISAC and OFDM-ISAC systems, which is analogous to the relation among CSE, BER and SINR for communications systems. Moreover, given the same SOP, SSE of our AFDM-ISAC system significantly outperforms the OFDM-ISAC system.
\end{example}

\begin{figure}[!htbp] 	
	%		\vspace*{-5pt} %留空白，可自己调整
	\centerline{\includegraphics[width=2.3 in]{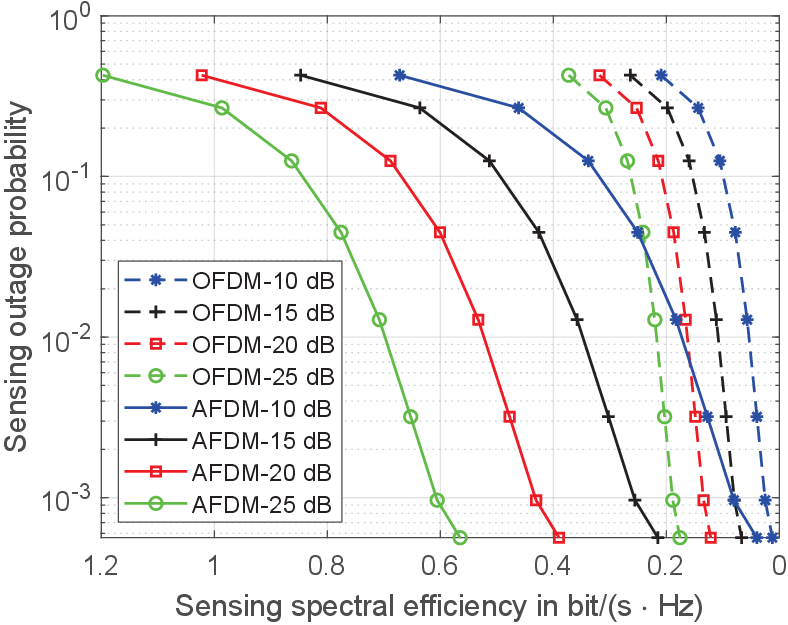}}
	\caption{The derived trade-off between SSE and SOP with different SINRs of AFDM-ISAC and OFDM-ISAC systems.}
	\label{fg:Pe_sen_vs_effi_sen}	
	%		\vspace*{-5pt} %留空白，可自己调整
\end{figure}

\begin{figure*}[htbp]
	\vspace*{-10pt} %留空白，可自己调整
	\begin{align} \label{eq:sum_effi_2}
		\eta _{sc}  = \left\{ {\begin{array}{*{20}{c}}
				{{k _{\eta}}\frac{N-1}{{N + {N_{cp}}}}{{\log }_2}{M_{mod}} + \left( {1 - {k _{\eta}}} \right)\left[ {\frac{{1 - 2{c_1}}}{{{c_1}}}\left( {{c_1} - \frac{{2{\xi _v} + 1}}{{2N}}} \right)I_{sen}} \right],}&{{c_1} > \frac{1}{{2\left( {{N_{cp}} + 1} \right)}}},\\
				{{k _{\eta}}\frac{N-1}{{N + {N_{cp}}}}{{\log }_2}{M_{mod}} + \left( {1 - {k _{\eta}}} \right)\left[ {2{N_{cp}}\left( {{c_1} - \frac{{2{\xi _v} + 1}}{{2N}}} \right)I_{sen} } \right],}&{{c_1} \le \frac{1}{{2\left( {{N_{cp}} + 1} \right)}}}.
		\end{array}} \right.
	\end{align}
	\hrule
	\vspace*{-10pt} %留空白，可自己调整
\end{figure*}

Moreover, according to (\ref{eq:t_m_AFDM}) and (\ref{eq:fd_m_AFDM}), the maximum tolerable delay ${\tau _{max}}$ is traded off against the maximum tolerable Doppler $f_{d,max}$ by changing parameter $c_1$. Specifically, the maximum tolerable Doppler is directly proportional to $c_1$, but the maximum tolerable delay is inversely proportional to $c_1$, when ${c_1} > \frac{1}{{2\left( {{N_{cp}}{\rm{ + }}1} \right)}}$. At this time, the relation between the maximum tolerable delay and Doppler for our AFDM-ISAC system, i.e., ${\tau _{max}^A}$ and $f_{d,max}^A$, is given by
\begin{align}\label{eq:fd_m_AFDM_2}
	{f_{d,max}^{A}} = \frac{B}{{2\left( {B{\tau _{max}^A} + 1} \right)}} - \frac{{B\left( {2{\xi _v} + 1} \right)}}{{2N}}.
\end{align}
When ${c_1} \le \frac{1}{{2\left( {{N_{cp}}{\rm{ + }}1} \right)}}$, the maximum tolerable Doppler is still directly proportional to $c_1$, but the maximum tolerable delay keep constant, for fixed $N_{cp}$.

On this basis, we analytically compare the trade-offs between ${\tau _{max}}$ and ${f_{d,max}}$ of AFDM-ISAC and OFDM-ISAC systems.
As a comparison, for OFDM-ISAC system, $\tau_{max}^O = T_{cp}$, and $f_{d,max}^O = {1 \mathord{\left/
		{\vphantom {1 {\left( {2{T_p}} \right)}}} \right.
		\kern-\nulldelimiterspace} {\left( {2{T_p}} \right)}}$, where $\tau_{max}^O$ and $f_{d,max}^O$ denote the maximum tolerable delay and Doppler, and $T_p$ is the duration of an OFDM symbol considering CP\cite{sturm2011waveform}. Consequently, the trade-off between ${\tau _{max}^O}$ and ${f_{d,max}^O}$ of OFDM-ISAC system is given by
\begin{equation}\label{eq:t_vs_fd_OFDM}
	{f _{d,max}^O} = \frac{{\left(1-{\eta _{cp}}\right)}}{{2{\tau _{max}^O}}},
\end{equation}
where ${\eta _{cp}} = {{\left( {{T_p} - {T_{cp}}} \right)} \mathord{\left/
		{\vphantom {{\left( {{T_p} - {T_{cp}}} \right)} {{T_p}}}} \right.
		\kern-\nulldelimiterspace} {{T_p}}}={{N} \mathord{\left/
		{\vphantom {{N} {{\left(N+N_{cp}\right)}}}} \right.
		\kern-\nulldelimiterspace} {{\left(N+N_{cp}\right)}}}$.
To compare intuitively, We let $\eta _{com}^A = \eta _{com}^O$ by setting identical parameters $M$, $N$, $N_{cp}$ and $M_{mod}$. Moreover, we set ${c_1} = \frac{1}{{2\left( {{N_{cp}} + 1} \right)}}$ to have $\tau_{max}^A = \tau_{max}^O = {{{N_{cp}}} \mathord{\left/
		{\vphantom {{{N_{cp}}} B}} \right.
	\kern-\nulldelimiterspace} B}$. At this time, according to (\ref{eq:fd_m_AFDM_2}) and (\ref{eq:t_vs_fd_OFDM}), we have
\begin{align}\label{eq:fd_m_AFDM_3}
	{f_{d,max}^A} &= \frac{B}{N}\left( {N\frac{1}{{2\left( {{N_{cp}} + 1} \right)}} - {\xi _v} - \frac{1}{2}} \right) \nonumber \\	
	& \approx \frac{\left( {1 - {\eta _{cp}}} \right)}{{2{\tau_{max}^A}}} \cdot \frac{{\left( {2{\xi _v} + 2} \right){\eta _{cp}} - \left( {2{\xi _v} + 1} \right)}}{{{\eta _{cp}}\left( {1 - {\eta _{cp}}} \right)}} \nonumber \\	
	& =  {f _{d,max}^O} \cdot {k_0},  
\end{align}
where $k_0 = {{\left( {2{\xi _v} + 2} \right){\eta _{cp}} - \left( {2{\xi _v} + 1} \right)} \mathord{\left/
		{\vphantom {{\left( {2{\xi _v} + 2} \right){\eta _{cp}} - \left( {2{\xi _v} + 1} \right)} {{\eta _{cp}}\left( {1 - {\eta _{cp}}} \right)}}} \right.
		\kern-\nulldelimiterspace} {{\eta _{cp}}\left( {1 - {\eta _{cp}}} \right)}}$. If ${\eta _{cp}} > \frac{1}{2}\left[ {\sqrt {\left( {2\xi {}_v + 1} \right)\left( {2\xi _v + 5} \right)}  - \left( {2\xi {}_v + 1} \right)} \right]$, we can get $k_0>1$ and thus
\begin{equation}\label{eq:fd_m_AFDM_4}	
{f_{d,max}^A} = {f _{d,max}^O} \cdot {k_0} > {f _{d,max}^O}.
\end{equation}
\begin{example}
When parameters are the same as Example \ref{ex:SSE_vs_SOP}, it can get ${\eta _{cp}} = {{8192} \mathord{\left/
{\vphantom {{8192} {\left( {8192 + 574} \right)}}} \right.
\kern-\nulldelimiterspace} {\left( {8192 + 574} \right)}} = 0.934$, and then $k_0 = 5.6$; that is, ${f_{d,max}^A}$ is 5.6 times ${f_{d,max}^O}$, while $\tau_{max}^A = \tau_{max}^O$ and $\eta_{com}^A = \eta_{com}^O$. According to (\ref{eq:eta_r_AFDM}), it results that SSE of the AFDM-ISAC system is also 5.6 times that of the OFDM-ISAC system.
\end{example}
	
\vspace{1ex}
\noindent\emph{(2) Trade-offs between performaces of sensing and communications}
\vspace{1ex}

Next, the trade-offs between performance metrics of sensing and communications is revealed, i.e., the trade-off between (\romannumeral1) CSE and ${\tau_{max}}$, (\romannumeral2) CSE and SSE.
According to (\ref{eq:t_m_AFDM}) and (\ref{eq:eta_c_AFDM}), CSE is traded-off against the maximum tolerable delay ${\tau_{max}}$ by changing parameters $N_{cp}$, when ${c_1} \le \frac{1}{{2\left( {{N_{cp}} + 1} \right)}}$. Specifically, increasing $N_{cp}$ enlarges the maximum tolerable delay but decreases CSE, for the fixed $N$. The relation between ${\eta _{com}^A}$ and ${\tau_{max}^A}$ for our AFDM-ISAC system can be expressed as
\begin{equation}
{\eta _{com}^A} = \frac{{\left(N-1\right){{\log }_2}{M_{mod}}}}{{N + B{\tau _{max}^A}}}.
\end{equation}

On this basis, we can get that CSE is traded-off against SSE by changing parameters $N_{cp}$, when ${c_1} \le \frac{1}{{2\left( {{N_{cp}} + 1} \right)}}$. The relation between ${\eta _{com}^A}$ and ${\eta _{sen}^A}$ for our AFDM-ISAC system is given by
\begin{equation}
{\eta _{com}^A} = \frac{{\left(N-1\right){{\log }_2}{M_{mod }}}}{{N + {{N{\eta _{sen}^A}} \mathord{\left/
				{\vphantom {{N{\eta _{sen}^A}} {\left\{ {\left[ {2N{c_1} - \left( {2{\xi _v} + 1} \right)} \right]{I_{sen}}} \right\}}}} \right.
				\kern-\nulldelimiterspace} {\left\{ {\left[ {2N{c_1} - \left( {2{\xi _v} + 1} \right)} \right]{I_{sen}}} \right\}}}}}.
\end{equation}
When ${c_1} \ge \frac{1}{{2\left( {{N_{cp}} + 1} \right)}}$, CSE decreases with increasing $N_{cp}$, but SSE keep contant.

\subsection{Design Guideline on Selecting AFDM Parameters}

Since the sensing and communication performances are strongly dependent on AFDM parameters $c_1$, $N_{cp}$ and $N$, we then provide a design guideline on selecting these AFDM parameters. Specifically, utilizing the derived analytical relationship in Eq. (\ref{eq:t_m_AFDM}) and Eq. (\ref{eq:fd_m_AFDM}), for the given requirements of maximum tolerable delay and Doppler, i.e., ${\tau _{max}}$ and ${f_{d,max}}$, the selections of $c_1$, $N_{cp}$ and $N$ should be bounded by 
\begin{align} \label{eq:cons}
\frac{{{f_{d,max}}}}{{B}} + \frac{{2{\xi _v} + 1}}{{2N}} &\le {c_1} \le \frac{1}{{2\left(B{\tau _{max}} + 1\right)}}, \\
\frac{{\left( {2{\xi _v} + 1} \right)\left( {B{\tau _{max}} + 1} \right)B}}{{B -  2{f_{d,max}}\left( {B{\tau _{max}} + 1} \right)}} &\le N  ,    \tag{\ref{eq:cons}a} \nonumber \\
{{B{\tau _{max}}}} &\le  {N_{cp}} \le N, \tag{\ref{eq:cons}b}
\end{align}
where (\ref{eq:cons}a) is obtained from the constraint that the upper bound should be larger than the lower bound in (\ref{eq:cons}).

\section{Estimation Method of Target Parameters for AFDM-ISAC System}

%\subsection{The AFDM-Based ISAC System Framework}

In this section, an efficient estimation method for our AFDM-ISAC system is proposed to estimate the delay and the Doppler of sensing targets in the AFT-Doppler domain. We start from deriving the input-output relationship of the AFDM-ISAC system in the AFT-Doppler domain. Then, the delay and the integral/fractional parts of normalized Doppler are extracted in the AFT-Doppler domain. 

\subsection{Input-Output Relationship of AFDM-ISAC System in the AFT-Doppler Domain}

Considering a sensing channel with $P$ point-like targets, the time delay and Doppler shift of the $i$-th target are denoted by ${\tau _i}$, ${f_{d,i}}$, respectively. Following (\ref{eq:received_sig_time}), after passing through the sensing channel, the received ISAC waveform in the time domain is given by \cite[Eq. (6)]{wu2022integrating}
\begin{equation}
	%\mathbf{\tilde r}\left[ n \right] = \sum\nolimits_{i = 1}^P {{{\tilde h}_i}} \mathbf{\tilde s}\left[ {n - {l_i}} \right]{e^{j2\pi {f_i}n}} + {\mathbf{\tilde w}_r}\left[ n \right],
	\mathbf{r}_\mathrm{I}\left[ n \right] = \sum\nolimits_{i = 1}^P {{{\chi}_i\left[n\right]}{e^{ - j2\pi {f_{d,i}}{\tau _i}}}} \mathbf{s}_\mathrm{I}\left[ {n - {l_i}} \right]{e^{j2\pi {f_i}n}} + {\mathbf{w}}_\mathrm{s}\left[ n \right],
\end{equation}
where $n \in \left[ { - {N_{cp}},N - 1 + \left( {{N_{sym}} - 1} \right){N_s}} \right]$, and $\mathbf{{w}}_\mathrm{s}\sim \mathcal {CN}\left( {0\text{,} \sigma_{s} ^2\mathbf{I}} \right)$ is an additive Gaussian noise vector. { The value of ${\chi}_i$ denotes the gain coefficient of the sensing channel corresponding to the $i$-th target, which is affected by the radar cross section (RCS) of the target. Due to the fluctuation of the target, the RCS is a random variable, and the gain coefficient ${\chi}_i$ may be changed at each sample. Following \cite[Chap. 7]{mark2010principles}, this paper models the gain coefficient ${\chi}_i$ as the Swerling fluctuation model, whose pdf is given by\cite{mark2010principles}
	\begin{equation}
	f\left( \chi  \right) = \frac{m}{{\left( {m - 1} \right)!\bar \chi }}{\left( {\frac{{m\chi }}{{\bar \chi }}} \right)^{m - 1}}{e^{\left( { - \frac{{m\chi }}{{\bar \chi }}} \right)}},
	\end{equation}
	where ${\bar \chi }$ deontes the mean of $\chi$, and the degree of the Chi-distribution is $2m$.

\begin{remark}
	This paper considers two types of Swerling models, i.e., the Sweiling 0 model and the Swerling 3 model\cite{mark2010principles}. The former indicates that the target is nonfluctuating, thus $\chi\left[n\right] = \bar \chi, \forall n$. The latter corresponds to $m=2$ and the pdf of $\chi$ in the Swerling 3 model is rewritten as\cite{mark2010principles}
	\begin{equation}
	f\left( \chi  \right) = \frac{4\chi}{{\bar \chi }^2}{e^{\left( { - \frac{{2\chi }}{{\bar \chi }}} \right)}}.
	\end{equation}
\end{remark}
}

After serial to parallel conversion and discarding CPP, the received signal matrix in the delay-time domain is
\begin{align}
	&\mathbf{R}\left[ {n,k} \right] {=} \frac{1}{\sqrt{N}} \sum\limits_{i = 1}^P {{{\chi}_i\left[n+kN_s\right]}{e^{ - j2\pi {f_{d,i}}{\tau _i}}}}{e^{j2\pi {f_i}\left( {n{\rm{ + }}kN_s} \right)}}  \times  \nonumber \\
	&\sum\limits_{m=0}^{N-1} { {\mathbf{X}\left[ m,k \right]} {e^{j2\pi \left( {{c_1}{\left(n-l_i\right)^2}{\rm{ + }}{c_2}{m^2} + \left(n-l_i\right)m/N} \right)}} } 	{+} {\mathbf{W}}_\mathrm{t}\left[ {n,k} \right],
	%\mathbf{R}\left[ {n,k} \right] = \sum\nolimits_{l = 0}^\infty  {\mathbf{S}\left[ {n - l,k} \right]} {g^{n,k}}\left( l \right) + w,
\end{align}
where ${\mathbf{W}}_\mathrm{t}$ denotes noise matrix in the delay-time domain, $n = 0, \ldots ,N {-} 1$, and $k = 0, \ldots ,N_{sym} {-} 1$. 

\begin{remark}
	Since the gain coefficient $\chi _i$ is a random variable in the Swerling 3 model, it is hard to reveal the deterministic input-output relation in the AFT domain. To this end, we first derive the deterministic input-output relation in the AFT domain based on the Swerling 0 model, i.e., $\chi _i\left[n\right] = \bar{\chi} _i, \forall n$, such that the primary properties can be discovered and the corresponding parameter estimation method can be designed. Then, the simulations will consider both the Swerling 0 and Swerling 3 models to verify the applicability of the proposed system under the Swerling 3 models. Numerical results will show that our proposed estimation system can also work under the Swerling 3 model with a slight performance loss. 	
\end{remark}

Motivated by \cite[Eq. (26)]{bemani2023affine},
performing $N$ points DAFT on each column of $\mathbf{R}$, the input-output relationship of our AFDM-ISAC system in the AFT-time domain under the Swerling 0 model can be written in the matrix form as
\begin{equation} \label{eq:in_out_relation_1}
	\mathbf{Y} =\sum\nolimits_{i = 1}^P {{{\bar{\chi} _i}{e^{ - j2\pi {f_{d,i}}{\tau _i}}}}{\mathbf{H}_{\mathrm{A},i}}\mathbf{X}\mathbf{D}_i} + \mathbf{W}_\mathrm{a},
\end{equation}
where $\mathbf{D}_i= \mathrm{diag}\left( {{e^{ j2\pi {f_i}N_sk}}},k {=} 0, \ldots ,{N_{sym}} {-} 1 \right)$, $\mathbf{W}_\mathrm{a}=\mathbf{A}\mathbf{W}_\mathrm{t}$, and  ${\mathbf{H}_{\mathrm{A},i}}\left[ {p,q} \right] = \frac{1}{N}{e^{j\frac{{2\pi }}{N}\left( {N{c_1}l_i^2 - ql_i + N{c_2}\left( {{q^2} - {p^2}} \right)} \right)}}\mathcal{F}_i\left( {p,q} \right)$ with ${\mathcal{F}_i}\left( {p,q} \right) {=} \frac{{{e^{ - j2\pi \left( {p - q - {\nu _i} + 2N{c_1}{l_i}} \right)}} - 1}}{{{e^{ - j\frac{{2\pi }}{N}\left( {p - q - {\nu _i} + 2N{c_1}{l_i}} \right)}} - 1}}$. 

Following {\cite{bemani2023affine}, $c_2$ is set to be a rational number sufficiently smaller than $\frac{1}{{2N}}$. Thus, if $c_2$ is small enough, the value of $N{c_2}\left( {{{q_i}^2} - {p^2}} \right)$ will approach zero. For the integral normalized Doppler shift case, i.e., $\nu _i$ is integer, the input-output relationship in the AFT-time domain shown in Eq. (\ref{eq:in_out_relation_1}) can be rewritten as\footnote{While the parameter estimation method is derived utilizing the input-output relation in the integer normalized Doppler shift case, it is also suitable for the fractional normalized Doppler shift case.} 
\begin{align}\label{eq:in_out_relation_3}
	\mathbf{Y}\left[ {p,k} \right] \approx & \sum\nolimits_{i = 1}^P {\zeta _i}{e^{ j2\pi {f_i}\left( {N {+} {N_{cp}}} \right)k}}{e^{ - j\frac{{2\pi }}{N}{l_i}p}} \times \nonumber\\
		&	\mathbf{X}\left[ {{{\left\langle {p + lo{c_i}} \right\rangle}_N},k} \right] 
{+} \mathbf{W}_\mathrm{a}\left[ {p,k} \right],
\end{align}
where ${\zeta _i} = {{\bar{\chi} _i}{e^{ - j2\pi {f_{d,i}}{\tau _i}}}}{e^{j\frac{{2\pi }}{N}\left( {N{c_1}l_i^2} - l_i lo{c_i} \right)}}$.  

We can see from (\ref{eq:in_out_relation_3}) that there is a linear phase shift along the $k$-axis (time domain), which is caused by the  Doppler shift. Meanwhile, there are the linear phase shift and the cyclic shift of information symbols $\mathbf{X}$ along the $p$-axis (AFT domain). The former is caused by the delay $l_i$, and the latter is caused by the delay $l_i$ and the integral part of normalized Doppler shift, i.e., ${\alpha _i}$. Due to this couple of the linear phase shift and the cyclic shift along the AFT domain, the PSLR of the radar image obtained by algorithms in \cite{sturm2011waveform,zeng2020joint} will be severely deteriorated when the Doppler shift is significant. Numerical results will verify this conclusion. 

The effect of delay and the integral part of normalized Doppler shift in the AFT domain can be decoupled by actively compensating. According to the assumption that $N_{cp} > l_{max}$, the delay $l$ is satisfied $0 \le {l} < {N_{cp}}$. For each delay $l$, we generate a compensation matrix $\mathbf{L}^l = \mathrm{diag}\left( {{e^{j\frac{{2\pi }}{N}p{l}}},p {=} 0, \ldots ,N - 1} \right)$ and multiply it by $\mathbf{Y}$, i.e., $\mathbf{Y}^l_\mathrm{c} = \mathbf{L}^l\mathbf{Y}$ and 
\begin{align}\label{eq:compensation}
	\mathbf{Y}^l_\mathrm{c}\left[ {p,k} \right] = &\sum\limits_{i = 1}^P {{\zeta _i}{e^{ j2\pi {f_i}N_sk}}{e^{j\frac{{2\pi }}{N}\left( {l - {l_i}} \right)p}}\mathbf{X}\left[ {{{\left\langle {p + lo{c_i}} \right\rangle}_N},k} \right]}\nonumber\\
	&{+} \mathbf{W}_\mathrm{c}^{l}\left[ {p,k} \right], \ l {=} 0, \ldots ,N_{cp} - 1.
\end{align} 
where $\mathbf{W}^{l}_\mathrm{c} = \mathbf{L}^l \mathbf{W}_\mathrm{a}$. This process results in $N_{cp}$ matrices. 

Then, the cyclic shift value of information symbols $\mathbf{X}$ in matrix $\mathbf{Y}_\mathrm{c}^l$ can be estimated by performing the matched filter in the AFT domain. The resulting matrix can be written as
\begin{align}
	\mathbf{Z}^{l} \left[p,k\right] &= { \sum\limits_{n = 0}^{N - 1} {{\mathbf{X}^*}\left[ { - n,k} \right]\mathbf{Y}_\mathrm{c}^l\left[ {p - n},k \right]}} \nonumber \\
	&= {\mathbf{F}^{\rm{H}}}\left( {{{\left( {\mathbf{F}\mathbf{Y}^l_\mathrm{c}} \right)}} \odot \left( {\mathbf{F}\mathbf{X}} \right)^*} \right)\nonumber \\
	&= \sum\limits_{i = 1}^P {{{\tilde \zeta} _i}{e^{ j2\pi {f_i}N_sk}}\mathbf{\mu}_k\left[{{{\left\langle {p + lo{c_i}} \right\rangle}_N}},{ {\left( {l - {l_i}} \right)}}\right]} \nonumber \\
	& \quad\quad\quad {+} \mathbf{\tilde W}^{l}\left[p,k\right], 
\end{align} 
where ${\tilde \zeta} _i = {\bar{\chi} _i}{e^{ - j2\pi {f_{d,i}}{\tau _i}}}{e^{j\frac{{2\pi }}{N}\left( {N{c_1}l_i^2} -  lo{c_i} l \right)}}$, $\mathbf{\tilde W}^{l}\left[p,k\right] = \frac{1}{\sqrt{N}} \sum\nolimits_{n = 0}^{N - 1} {\mathbf{W}^{l}_\mathrm{c}\left[ {n,k} \right]{\mathbf{X}^ * }\left[ {{\left\langle p-n \right\rangle}_N,k} \right]}$, and { $\mathbf{\mu}_{k}\left[p,m\right] =  \frac{1}{\sqrt{N}}\sum\nolimits_{n = 0}^{N - 1} {\mathbf{X}\left[ {n,k} \right]{\mathbf{X}^ * }\left[ {{\left\langle n-p \right\rangle}_N,k} \right]e^{j\frac{{2\pi }}{N}mn}}$} denoting the periodic ambiguity function of the random symbol vector $\mathbf{X}\left[:,k\right]$. 
Then, performing the $N_{sym}$-point DFT at each row of $\mathbf{Z}^l$, we can get the matrix in the AFT-Doppler domain as (\ref{eq:in_out_rela_AFT_Z_3}) on the next page. In (\ref{eq:in_out_rela_AFT_Z_3}), $\mathbf{\tilde W}_\mathrm{F}^{l} = \mathbf{\tilde W}^{l}\mathbf{F}$, ${{\nu }_i'} = \left(N+N_{cp}\right)f_i = \frac{{{f_{d,i}}}}{{\Delta f'}} {=} {\beta _i} {+} {b_i}  {\in} \left[ { - {\nu' _{\max }}, {\nu' _{\max }}} \right]$ denoting the Doppler shift normalized with respect to ${\Delta f'}={B \mathord{\left/
			{\vphantom {B {{\left(N+N_{cp}\right)}}}} \right.
			\kern-\nulldelimiterspace} {{\left(N+N_{cp}\right)}}}$, and ${\beta _i} {\in} \left[ { - {\beta _{\max }},{\beta _{\max }}} \right]$ and ${b_i} {\in} \left( { - \frac{1}{2},\frac{1}{2}} \right]$ are the integral/fractional part of ${{\nu }_i'}$.
\begin{figure*}
	{
	\begin{align} \label{eq:in_out_rela_AFT_Z_3}
	\mathbf{Z}^{l}_\mathrm{F} \left[p,k\right] &= \mathbf{Z}^{l}\mathbf{F} = {\mathbf{F}^{{H}}}\left( {{{\left( {\mathbf{F}\mathbf{Y}^l_\mathrm{c}} \right)}} \odot \left( {\mathbf{F}\mathbf{X}} \right)^*} \right)\mathbf{F} \nonumber \\
	&= \frac{1}{{\sqrt {N{N_{sym}}} }}\sum\limits_{i = 1}^P {{{\tilde \zeta }_i}\sum\limits_{k = 0}^{{N_{sym}} - 1} {\sum\limits_{m = 0}^{N - 1} {X\left[ {m,k} \right]{X^ * }\left[ {{{\left\langle {m - \left( {p + lo{c_i}} \right)} \right\rangle }_N},k} \right]{e^{j\frac{{2\pi }}{N}\left( {l - {l_i}} \right)m}}} {e^{j\frac{{2\pi }}{{{N_{sym}}}}\left( {{N_{sym}}{{\nu }_i'} - n} \right)k}}} } + \mathbf{\tilde W}_\mathrm{F}^{l}\nonumber \\
	&= \frac{1}{{\sqrt {N{N_{sym}}} }}\sum\limits_{i = 1}^P {{{\tilde \zeta }_i}\sum\limits_{k = 0}^{{N_{sym}} - 1} {\mathbf{\mu}_k\left[{{{\left\langle {p + lo{c_i}} \right\rangle}_N}},{ {\left( {l - {l_i}} \right)}}\right] {e^{j\frac{{2\pi }}{{{N_{sym}}}}\left( {{N_{sym}}{{\nu }_i'} - n} \right)k}}} } + \mathbf{\tilde W}_\mathrm{F}^{l}.
	\end{align} }
	\hrule
	\vspace*{-10pt} %留空白，可自己调整
\end{figure*}

{ Following \cite{zeng2020joint}, for a given matrix $\mathbf{X}$ with random communications symbols, the peak of matrix $\left| {\mathbf{\mu}_{k}} \right|$ is $\left| {\mathbf{\mu}_{k}}\left[0,0\right] \right|$. Other cross terms (i.e., sidelobe) $\mathbf{\mu}_k\left[p,m\right]$ ($p \ne 0$, $m \ne 0$) is small compared with the peak when $N$ is large enough. This paper mainly focuses on finding the position of the peak and thus ignores the effects of side lobes when deriving the peak position following \cite{zeng2020joint}, i.e., $\mathbf{\mu}_k\left[p,m\right] \approx c_0 \delta\left[p\right]\delta\left[m\right]$. At this time, Eq. (\ref{eq:in_out_rela_AFT_Z_3}) can be approximately expressed as}

\begin{align} \label{eq:in_out_rela_AFT_Z_4}
	&\mathbf{Z}^{l}_\mathrm{F} \left[p,k\right]  \approx \frac{1}{\sqrt{N_{sym}}} \sum\limits_{i = 1}^P {{\tilde \zeta} _i}c_0\delta\left[{{\left\langle {p + lo{c_i}} \right\rangle}_N}\right]\delta\left[{{\left( {l - {l_i}} \right)}}\right]  \nonumber \\
	& \ \ \ \  \ \ \ \ \ \ \ \  \times	\frac{{{e^{ - j2\pi \left( {k - {N_{sym}}{{\nu }_i'}} \right)}} - 1}}{{{e^{{{ - j2\pi \left( {k - {N_{sym}}{{\nu }_i'}} \right)} \mathord{\left/
								{\vphantom {{ - j2\pi \left( {k - {N_{sym}}{{\nu }_i'}} \right)} {{N_{sym}}}}} \right.
								\kern-\nulldelimiterspace} {{N_{sym}}}}}} - 1}} {+} \mathbf{\tilde W}_\mathrm{F}^{l}\left[p,k\right]  \nonumber \\
	&  = \left\{ {\begin{array}{*{20}{c}}
			{\sum\limits_{i = 1}^P {{{{\tilde \zeta} }_i}{c_0}\sqrt{N_{sym}}}{+} \mathbf{\tilde W}_\mathrm{F}^{l}\left[p,k\right] ,}&{C. 1,}\\
			{\mathbf{\tilde W}_\mathrm{F}^{l}\left[p,k\right],}&{otherwise,}
	\end{array}} \right.
\end{align} 
where $p = 0, \ldots ,N {-} 1$, $k = 0, \ldots ,N_{sym} {-} 1$, and $C. 1$ denotes the constraint that $l=l_i \ {\rm and} \ $ ${{\left\langle {p + lo{c_i}} \right\rangle}_N} =0   \ {\rm and} \ {\left\langle k - {N_{sym}}{{\nu }_i'}\right\rangle}_{N_{sym}} = 0 $. Since ${\nu }_i' = {\beta _i} {+} {b_i}$, it can get ${\left\langle k - {N_{sym}}{{\nu }_i'}\right\rangle}_{N_{sym}} = {\left\langle k - {N_{sym}}{{b }_i}\right\rangle}_{N_{sym}}$. { Eq. (\ref{eq:in_out_rela_AFT_Z_4}) reveals that $P$ peaks will occur if and only if the constraint $C. 1$ holds. In other words, there will be $P$ peaks at the indices of ${\bar p_i}$ and ${\bar k_i}$ in the matrix $\mathbf{Z}^{\bar l_i}_\mathrm{F}$, where
$\bar l_i = l$, ${\bar p_i} = {{\left\langle {{\alpha _i} - 2N{c_1}{l_i} } \right\rangle}_N}$ and ${\bar k_i = {\left\langle {N_{sym}}{{b }_i}  \right\rangle}_{N_{sym}}}$.}

While this conclusion is obtained by ignoring the cross terms of $\mathbf{\mu}_k$, it still holds when the cross-term exists verified by numerical results. It means that the peak indices contain information on the delay and the integral/fractional parts of the normalized Doppler of targets, which enables us to design the following parameter estimation method. 

\subsection{Proposed Parameter Estimation Method of Sensing Target}

Accordingly, we propose an estimation algorithm to estimate the delay and the Doppler of sensing targets. First, according to the above analyses, the received signal matrix $\mathbf{Y}$ in the AFT-time domain is multiplied by $N_{cp}$ compensation matrices $\mathbf{L}^l$, resulting in $N_{cp}$ matrices $\mathbf{Y}^l_\mathrm{c}$. The matched filter in the AFT domain and DFT in the time domain are performed on matrices $\mathbf{Y}^l_\mathrm{c}$, resulting in $N_{cp}$ matrices $\mathbf{Z}^{l}_\mathrm{F}$ in the AFT-Doppler.

Then, we use the classical constant false alarm rate (CFAR) algorithm to search the indices of the peak in matrices $\mathbf{Z}^{l}_\mathrm{F}$. 
If the magnitude of the peak at the indices of ${\bar p_i}$ and ${\bar k_i}$ in the $\bar l_i$-th matrix $\mathbf{Z}^{\bar l_i}_\mathrm{F}$ exceeds the threshold, there is a target whose corresponding normalized delay estimation is given by $\hat l_i = \bar l_i$, and thus the delay estimation is given by
\begin{eqnarray}\label{eq:esti_delay}
	\hat \tau _i = \frac{1}{B} \bar l_i .
\end{eqnarray}
According to the constraint that ${\bar p_i} = {{\left\langle {{\alpha _i} - 2N{c_1}{l_i} } \right\rangle}_N}$, the integral part of normalized Doppler shift is estimated by
\begin{equation}\label{eq:int_Doppler}
	%	{\hat \alpha _i} = 2N{c_1}{\hat l_i} {-} {\left( {1 - {{\bar p}_i}} \right)_N}.
	{\hat \alpha _i} =  {\left\langle {{{\bar p}_i} + 2N{c_1}{{\hat l}_i} + N{c_1}} \right\rangle _N} - N{c_1}.
\end{equation}
Moreover, according to the equality of ${\bar k_i = {\left\langle {N_{sym}}{{b }_i}  \right\rangle}_{N_{sym}}}$, the estimation of the fractional part of normalized Doppler shift, i.e., ${b_i}$, is given by
	\begin{equation}\label{eq:frac_Doppler}
					{{\hat b}_i} = \frac{{{{{N_{sym}}} \mathord{\left/
									{\vphantom {{{N_{sym}}} 2}} \right.
									\kern-\nulldelimiterspace} 2} - {{\left\langle {{{{N_{sym}}} \mathord{\left/
												{\vphantom {{{N_{sym}}} 2}} \right.
												\kern-\nulldelimiterspace} 2}-{{\bar k}_i} } \right\rangle}_{N_{sym}}} }}{{{N_{sym}}}}.
\end{equation} 
{ The range of ${{\hat b}_i}$ is $\left(-0.5, 0.5\right]$ for a given ${{\bar k}_i} \in \left[0,N_{sym}\right]$ in the Doppler domain. It means that the range of estimated Doppler is $\left(-0.5{\Delta f'}, 0.5{\Delta f'}\right]$ from the Doppler domain.} 

If the integral and fractional parts of normalized Doppler can be spliced, the unambiguous Doppler will break through the limitations of subcarrier spacing $\Delta _f$. Note that the values of $\alpha _i$ and ${b_i}$ cannot be directly added because they are the integral and fractional parts of two different normalized Doppler shifts $\nu _i$ and $\nu _i'$, respectively. According to the definitions of $\alpha _i$ and ${b_i}$, we can get 
\begin{equation}\label{eq:fd_esti}
	{{\hat f}_{d,i}} = \left( {\hat \alpha _i  + \hat a _i} \right)\Delta f = \left( {\hat \beta _i + \hat b _i} \right)\Delta f',
\end{equation}
and
\begin{align}\label{eq:a_1}
%	\hat a _i &= \frac{{\Delta f'}}{{\Delta f}}\left( {\hat \beta _i + \hat b _i} \right) - \hat \alpha_i  \nonumber\\
%	& = \frac{N}{{N + {N_{cp}}}}\hat \beta _i + \frac{N}{{N + {N_{cp}}}}\hat b _i- \hat \alpha _i.	
	\hat a _i = \frac{N}{{N + {N_{cp}}}}\hat \beta _i + \frac{N}{{N + {N_{cp}}}}\hat b _i- \hat \alpha _i.
\end{align}
We can see that $\hat a _i$ is a function of $\hat \beta _i$.
According to the condition ${a_i} \in \left( { - \frac{1}{2},\frac{1}{2}} \right]$, we can get $\hat \beta_i  \in \left( {\frac{{N + {N_{cp}}}}{N}\hat \alpha_i  - \hat b_i - 0.5\frac{{N + {N_{cp}}}}{N},\frac{{N + {N_{cp}}}}{N}\hat \alpha_i  - \hat b_i + 0.5\frac{{N + {N_{cp}}}}{N}} \right]$. The minimum and maximum values of $\hat \beta_i$ are ${{{\hat \beta }_{i,\min }} = \left\lceil {\frac{{N + {N_{cp}}}}{N}\hat \alpha_i  - \hat b_i - 0.5\frac{{N + {N_{cp}}}}{N}} \right\rceil }$ and ${{\hat \beta }_{i,\max }} = \left\lfloor {\frac{{N + {N_{cp}}}}{N}\hat \alpha_i  - \hat b_i + 0.5\frac{{N + {N_{cp}}}}{N}} \right\rfloor $, respectively. When $0 \le {N_{cp}} < N$, there are only two cases, either ${{\hat \beta }_{i,\max }} {=} {{\hat \beta }_{i,\min }}$ or ${{\hat \beta }_{i,\max }} {=} {{\hat \beta }_{i,\min }} {+} 1$. If ${{\hat \beta }_{i,\max }} {=} {{\hat \beta }_{i,\min }}$, let ${{\hat \beta_i }} {=} {{\hat \beta }_{i,\max }}$ and substitute it into (\ref{eq:a_1}) to compute $\hat a _i$. If ${{\hat \beta }_{i,\max }} {=} {{\hat \beta }_{i,\min }} {+} 1$, we introduce an early-late criterion to judge whether ${{\hat \beta_i }}$ should be equal to ${{\hat \beta }_{i,\max }}$ or ${{\hat \beta }_{i,\min }}$. The early-late judgment criterion is given by
\begin{equation}\label{eq:judgment}
	\hat \beta_i  = \left\{ {\begin{array}{*{20}{c}}
			{{{\hat \beta }_{i,\min }},}&\hspace{-2ex}\left|{\mathbf{Z}^{\hat l_i}_\mathrm{F}}\left[ {{{\bar p}_i} - 1,{{\bar k}_i}} \right]\right| > \left|{\mathbf{Z}^{\hat l_i}_\mathrm{F}}\left[ {{{\bar p}_i} + 1,{{\bar k}_i}} \right]\right|,\\
			{{{\hat \beta }_{i,\max }},}&{otherwise,}
	\end{array}} \right.
\end{equation}
where $\mathbf{Z}^{\hat l_i}_\mathrm{F}\left[ {{{\bar p}_i}{-}1,{{\bar k}_i}} \right]$ and $\mathbf{Z}^{\hat l_i}_\mathrm{F}\left[ {{{\bar p}_i}{+}1,{{\bar k}_i}} \right]$ are the early and late elements of the peak (i.e., $\mathbf{Z}^{\hat l_i}_\mathrm{F}\left[ {{{\bar p}_i},{{\bar k}_i}} \right]$) along the $p$-axis, respectively.  
{ Eq. (\ref{eq:judgment}) means that when the amplitude of the early element of the peak is larger than that of the late element, $\beta_i$ is estimated as $\hat \beta_i = {{\hat \beta }_{i,\min }}$. Otherwise, $\beta_i$ is estimated as $\hat \beta_i = {{\hat \beta }_{i,\max }}$.}
Finally, substituting the estimated $\hat \beta _i$ in (\ref{eq:judgment}) into (\ref{eq:fd_esti}), the Doppler shift of the $i$-th target, i.e., ${{\hat f}_{d,i}}$, can be estimated. The procedure to estimate the delay and the Doppler of targets is summarized as Algorithm \ref{algo:esti_delay_Doppler}.

\begin{algorithm}[htb]
	\caption{Estimating Delay and Doppler in the AFT-Doppler domain}\label{algo:esti_delay_Doppler}
	
	\textbf{Iuput}: $\mathbf{Y}$, $N$, $N_{cp}$, $N_{sym}$, $c_1$, $c_2$.
	
	\textbf{Output}: Estimated delay and Doppler $\hat \tau _i$ and ${{\hat f}_{d,i}}$.
	
	1: \ Compute $\mathbf{Y}^l_\mathrm{c}$ by (\ref{eq:compensation}), $l=0,\dots, N_{cp}-1$.
	
	2: \ Perform matched filter in the AFT domain and DFT to obtain $\mathbf{Z}^{l}_\mathrm{F}$ by (\ref{eq:in_out_rela_AFT_Z_3}).
	
	3: \ Extrate the indices of ${\bar l_i}$ ${\bar p_i}$ and ${\bar k_i}$ by searching peak. 
	
	4: \ Estimate $\hat l_i$, $\hat \tau _i$, ${\hat \alpha _i}$ and ${{\hat b}_i}$ by (\ref{eq:esti_delay}), (\ref{eq:int_Doppler}) and (\ref{eq:frac_Doppler}). 	
	
	5: \ Estimate ${{\hat f}_{d,i}}$ by (\ref{eq:fd_esti}), (\ref{eq:a_1}) and (\ref{eq:judgment}). 		
\end{algorithm}

{
\subsection{Computational Complexity Analysis}

In this subsection, we analyze the computational complexity of our proposed Algorithm \ref{algo:esti_delay_Doppler}. Following the concept of \cite{liu2018mu}, obtaining the matrix $\mathbf{Y}$ should perform $N_{sym}$ times $N$-point DAFT, which takes ${\mathcal {O}}\left(N_{sym}\left(N \log_2 N + N \right)\right)$ flops. Computing $\mathbf{Y}^l_\mathrm{c}, l=0,\dots, N_{cp}-1$ by (\ref{eq:compensation}) takes ${\mathcal {O}}\left(N_{cp}N_{sym}N \right)$ flops. Performing matched filter and DFT to obtain $\mathbf{Z}^{l}_\mathrm{F}$ by (\ref{eq:in_out_rela_AFT_Z_3}) takes ${\mathcal {O}}\left(N_{cp}N_{sym}N \log_2 N + N_{cp}NN_{sym}\log_2 N_{sym} \right)$ flops. Consequently, the computational complexity of our proposed Algorithm \ref{algo:esti_delay_Doppler} is ${\mathcal {O}}\left(N_{cp}NN_{sym}\log_2\left( NN_{sym} \right)\right)$ by ignoring the low-order terms. As a comparison, the complexity of the element-wise division method with low-complexity for OFDM-ISAC system is ${\mathcal {O}}\left(NN_{sym}\log_2\left( NN_{sym} \right)\right)$\cite{sturm2011waveform}. }

{ Finally, we briefly summarize how the AFDM chirp properties are leveraged in the proposed ISAC system. Thanks to the chirp properties, the delay and the integral part of normalized Doppler can be decoupled in the affine Fourier transform domain for the AFDM-ISAC system. We use this property to control the performance trade-offs of ISAC system, e.g., the maximum tolerable delay and the maximum tolerable Doppler. Moreover, we use this property to make the range of estimated Doppler breaking through the limitations of subcarrier spacing with an acceptable complexity.}

\section{Numerical Results}

In this section, we present the numerical results of ISAC systems.
This paper mainly focuses on comparing the performances of our AFDM-ISAC system and the existing OFDM-ISAC\cite{sturm2011waveform,zeng2020joint} and OTFS-ISAC\cite{gaudio2020effectiveness} systems. Unless otherwise specified, simulation parameters are listed in Table \ref{tab:parameter_simu}. { The subcarrier spacing of 48 kHz corresponds to a round-trip velocity of 300 m/s.}

\begin{table}
	\centering
	%%\caption{****}
	\caption{Simulations Paremeters}\label{tab:parameter_simu}
	\renewcommand\arraystretch{1.3}
	\begin{threeparttable}
		{
		\begin{tabular}{|c|l|c|}
			
			\hline
			\textbf{Symbol}                       &\textbf{Parameter}                       & \textbf{Value}     \\ \hline
			$f_c$ &  Carrier frequency                      & $24$ GHz     \\ \hline
			${B}$ &  Bandwidth                       & $122.88$ MHz  \\ \hline
			$\Delta _f$ &  Subcarrier spacing                       & $48$ kHz  \\ \hline
			$N$ &  Number of subcarriers                      & $2560$     \\ \hline
			$N_{cp}$ &  Number of chirp-periodic prefix (CPP)                      & $288$     \\ \hline
			$N_{sym}$ &  Number of AFDM symbols per frame                      & $32$     \\ \hline						
			$\nu _{\max}$ &  maximum normalized Doppler shift                       & 2  \\ \hline		
			${\bar \chi}_i$ &  Mean gain coefficient of sensing channel                       & 1  \\ \hline
%			$\Delta R$ &  Range resolution                      & $1.22$ m  \\ \hline		
		\end{tabular}
		}
	\end{threeparttable}
\end{table}

\subsection{Trade-offs Between S$\&$C Performances}

\begin{figure}[!htbp]
	\centering
	\subfigure[SSE vs. CSE]{
		\includegraphics[width=2.3in]{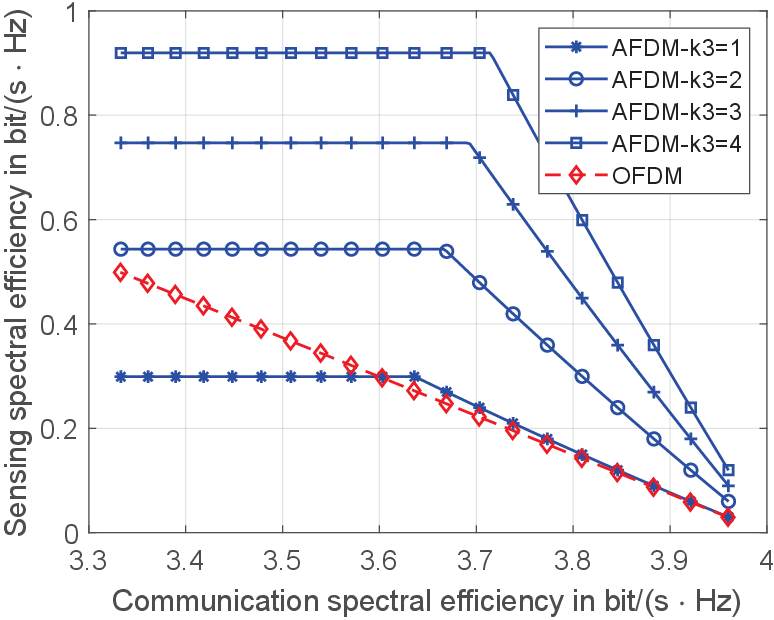}
	}
	\subfigure[Unambiguous range vs. unambiguous velocity]{
		\includegraphics[width=2.3in]{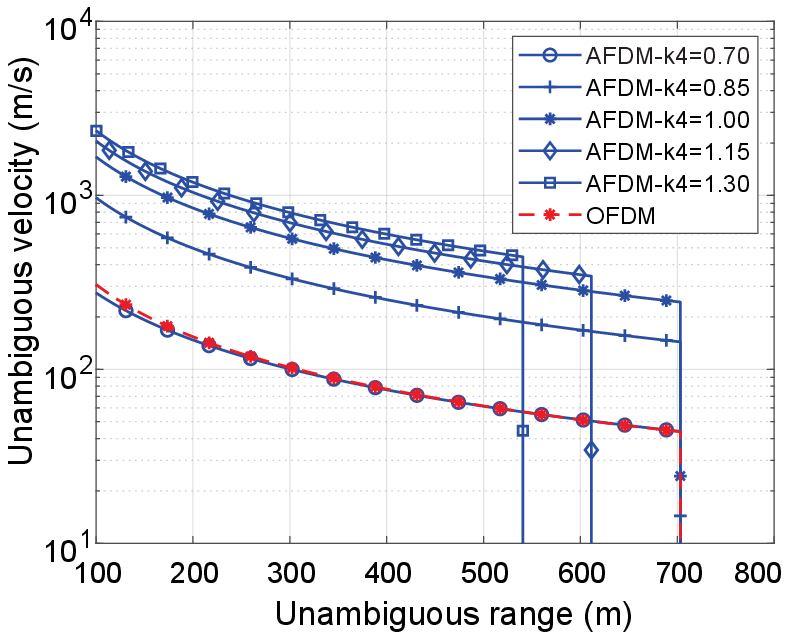}
	}
	%	\vspace*{-5pt} %留空白，可自己调整
	\caption{Trade-off between (a) SSE and CSE, (b) unambiguous range and unambiguous velocity of sensing, for AFDM-ISAC and OFDM-ISAC systems with different $c_1$.  
		\label{fg:trade_off}}
	%	\vspace*{-5pt} %留空白，可自己调整
\end{figure}

Firstly, we show the simulation results of the trade-offs between S$\&$C performances. The trade-offs between SSE and CSE of AFDM-ISAC and OFDM-ISAC systems with different $c_1$ and $N_{cp}$ are shown in Fig. \ref{fg:trade_off}(a), where $c_1=c_2=0$ for OFDM, and $c_1 = {{{\left(2\xi _v+1\right)}} \mathord{\left/
		{\vphantom {{{\left(2\xi _v+1\right)}} \left(2N\right)}} \right.
		\kern-\nulldelimiterspace} \left(2N\right)} + {{{k_3}} \mathord{\left/
		{\vphantom {{{k_3}} \left(2N\right)}} \right.
		\kern-\nulldelimiterspace} \left(2N\right)}$, $c_2=0$ and $\xi _v=4$ for AFDM, where $k_3=\left\{1,2,3,4\right\}$. $M_{mod} = 16$, $N=8192$, $\gamma = 0.5$ for both waveforms. The coverage radius of the ISAC base station $R_{max}$ is from 100 m to 2000 m, and $N_{cp} = \left\lceil\frac{2R_{max}B}{c}\right\rceil$. As $N_{cp}$ increases, CSE decreases from 3.96 bit/(s$\cdot$Hz) to 3.33 bit/(s$\cdot$Hz), and SSE increases from 0.12 bit/(s$\cdot$Hz) to 0.92 bit/(s$\cdot$Hz) for the OFDM-ISAC system. For the AFDM-ISAC system, the trade-off between SSE and CSE is strongly related to the parameter $c_1$. The larger the $c_1$, the higher the SSE. When $c_1 \ge {{{\left(2\xi _v+1\right)}} \mathord{\left/
		{\vphantom {{{\left(2\xi _v+1\right)}} \left(2N\right)}} \right.
		\kern-\nulldelimiterspace} \left(2N\right)} + {{{2}} \mathord{\left/
		{\vphantom {{{k_3}} \left(2N\right)}} \right.
		\kern-\nulldelimiterspace} \left(2N\right)}$, i.e., $k_3 \ge 2$, SSE of AFDM-ISAC system outperforms OFDM-ISAC system, given the same CSE. { Moreover, for a fixed $c_1$, SSE is a piecewise function of $N_{cp}$ according to Eq. (\ref{eq:eta_r_AFDM}), where SSE is proportional to $N_{cp}$ when $N_{cp} \le {{1} \mathord{\left/		{\vphantom {{1} \left(2c_1\right)}} \right. \kern-\nulldelimiterspace} \left(2c_1\right)}-1$, and then SSE keeps constant regardless of the increase of $N_{cp}$ when $N_{cp} > {{1} \mathord{\left/
			{\vphantom {{1} \left(2c_1\right)}} \right.
			\kern-\nulldelimiterspace} \left(2c_1\right)}-1$. CSE is inversely proportional to $N_{cp}$, which always holds on. As a result, as $N_{cp}$ increases, 
	SSE firstly increases with the decrease of CSE when $N_{cp} < {{1} \mathord{\left/
			{\vphantom {{1} \left(2c_1\right)}} \right.
			\kern-\nulldelimiterspace} \left(2c_1\right)}-1$ and then keeps constant with the decrease of CSE when $N_{cp} > {{1} \mathord{\left/
		{\vphantom {{1} \left(2c_1\right)}} \right.
		\kern-\nulldelimiterspace} \left(2c_1\right)}-1$ as shown in Fig. \ref{fg:trade_off}(a). The turning occurs at $N_{cp} = {{1} \mathord{\left/
		{\vphantom {{1} \left(2c_1\right)}} \right.
		\kern-\nulldelimiterspace} \left(2c_1\right)}-1$, and the coordinate of turning point in Fig. \ref{fg:trade_off}(a) is $\left( {\frac{{2\left( {N - 1} \right){c_1}{{\log }_2}{M_{mod }}}}{{2\left( {N - 1} \right){c_1} + 1}}, \frac{{1 - 2{c_1}}}{{{c_1}}}\left( {{c_1} - \frac{{2{\xi _v} + 1}}{{2N}}} \right)I_{sen}} \right)$.}

Figure \ref{fg:trade_off}(b) reveals the trade-off between unambiguous range and unambiguous velocity of AFDM-ISAC and OFDM-ISAC systems with different $c_1$ and $N$.
Two systems have the same waveform parameters following 5G NR ($N_{cp}=\frac{288}{4096}N$\cite{Techplayon20205G}), so that they have the same CSE. $\xi _v=4$, and $N$ is from 1024 to 8192 for both systems.
Benefitting from the flexibility of AFDM in parameter $c_1$, 
we set $c_1 = {{{k_4}} \mathord{\left/
		{\vphantom {{{k_4}} \left(2N_{cp}+2\right)}} \right.
		\kern-\nulldelimiterspace} \left(2N_{cp}+2\right)}$ and $k_4=\left\{0.7,0.85,1.0,1.15,1.3\right\}$. 
	We can observe that the unambiguous velocity decreases with the growth of unambiguous range for both systems, which is coincident with analytical relation in Eq. (\ref{eq:fd_m_AFDM_2}). For the AFDM-ISAC system, the trade-off between unambiguous range and unambiguous velocity is also strongly related to the parameter $c_1$. When $c_1={{{0.7}} \mathord{\left/
			{\vphantom {{{0.7}} \left(2N_{cp}+2\right)}} \right.
			\kern-\nulldelimiterspace} \left(2N_{cp}+2\right)}$, i.e., $k_4=0.7$, two systems have the same performances. With the growth of $c_1$ from ${{{0.7}} \mathord{\left/
			{\vphantom {{{0.7}} \left(2N_{cp}+2\right)}} \right.
			\kern-\nulldelimiterspace} \left(2N_{cp}+2\right)}$ to ${{{1.0}} \mathord{\left/
			{\vphantom {{{1.0}} \left(2N_{cp}+2\right)}} \right.
			\kern-\nulldelimiterspace} \left(2N_{cp}+2\right)}$, the unambiguous velocity is increased to five times that of OFDM-ISAC system, while maintaining the same CSE and unambiguous range.  
%		The reason for turning is similar to that in Fig. \ref{fg:trade_off}(a). 
		{ Moreover, the unambiguous range is also a piecewise function of $N_{cp}$ and $c_1$ according to Eq. (\ref{eq:t_m_AFDM}). For a given $c_1$, the unambiguous range is proportional to $N_{cp}$ when $c_1 \le {{{1}} \mathord{\left/
					{\vphantom {{{1}} \left(2N_{cp}+2\right)}} \right.
					\kern-\nulldelimiterspace} \left(2N_{cp}+2\right)}$ and will be bounded on ${\frac{c}{{2B}}\left( {\frac{1}{{2{c_1}}} - 1} \right)}$ when $c_1 > {{{1}} \mathord{\left/
					{\vphantom {{{1}} \left(2N_{cp}+2\right)}} \right.
					\kern-\nulldelimiterspace} \left(2N_{cp}+2\right)}$. Therefore, when $c_1 > {{{1}} \mathord{\left/
					{\vphantom {{{1}} \left(2N_{cp}+2\right)}} \right.
					\kern-\nulldelimiterspace} \left(2N_{cp}+2\right)}$, as $c_1$ increases, the upper bound on the unambiguous range will lost.}

\subsection{Simulation Results of Sensing Performances}

\begin{figure}[!htbp]
	\vspace*{-5pt} %留空白，可自己调整
	\centering
	\includegraphics[width=2.2in]{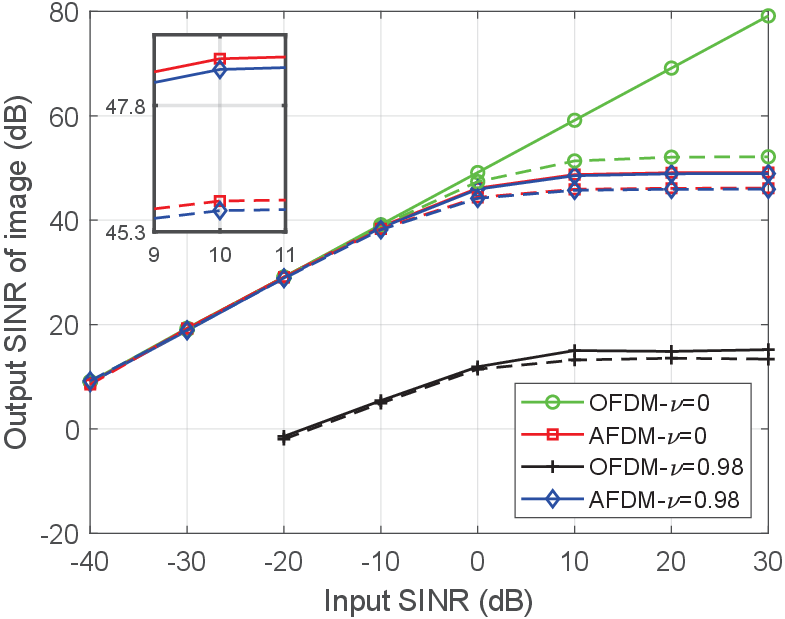}
	\caption{ Output SINR of image versus the input SINR of echo in both Swerling 0 and Swerling 3 models}
	\label{fg:Image_SNR_vs_inputSNR}
	\vspace*{-5pt} %留空白，可自己调整
\end{figure}

{ Figures \ref{fg:Image_SNR_vs_inputSNR} and \ref{fg:Image_SNR_vs_Doppler} show the sensing performances in both Swerling 0 and Swerling 3 models for AFDM-ISAC and OFDM-ISAC systems.}
In the simulation, AFDM and OFDM waveforms are set to the same waveform parameters following Table \ref{tab:parameter_simu} (except for $c_1$). $c_1=0.0025$ and $c_2=0$ for AFDM.
Figure \ref{fg:Image_SNR_vs_inputSNR} illustrates the simulation results of image signal-to-noise ratio (SNR), which is the ratio between the peak caused by the target and the average noise level in the two-dimensional radar image\cite{sturm2011waveform}, versus the SINR of received echoc. The average noise level of the radar image is simultaneously contributed by the input noise and the side lobes of the ambiguity function of the waveform. 
%{ The bandwidth $B=$ 122.88 MHz, $N=2560$, $N_{cp}=288$, $N_{sym}=32$, $c_1=0.0025$, $c_2=0$, $\xi _v=4$, $\Delta _f = 48$ kHz, and $M_{mod}=4$.} 
{ The solid and dashed lines represent the cases of the Swerling 0 and the Swerling 3 models, respectively. The Doppler shifts are set as $f_d{=}0$ and $f_d{=}0.98\Delta_f$, corresponding to the round-trip velocity of 0 m/s and 294 m/s, respectively.}

{ On the one hand, under the Swerling 0 model}, it can be observed that in the case of $f_d{=}0$, the available image SNRs decrease almost linearly with ${\rm SINR}$ for both systems when ${\rm SINR}$ is below 0 dB.
There appears to be a saturation of image SNR starting approximately at the ${\rm SINR}$ of 10 dB for our proposed AFDM-ISAC system. The OFDM-ISAC system can output higher image SNR at high ${\rm SINR}$ regions. This is because the OFDM-ISAC system can obtain lower side lobes by eliminating the effects of random communications symbols utilizing symbol division.
In the case of $f_d{=}0.98\Delta_f$, the saturations of image SNRs for both systems appear, and the image SNR of our proposed AFDM-ISAC system almost the same as that in the case of $f_d{=}0$. However, the image SNR severely decreases for the OFDM-ISAC system.
 
{ On the other hand, under the Swerling 3 model, the image SNR of OFDM-ISAC system in the cases of $f_d{=}0$ appears saturation in the high input SNR region. The reason is that the fluctuating gain efficient $\chi$ makes the received symbols in the frequency domain distorted, and thus the symbol division method in \cite{sturm2011waveform} cannot completely eliminate the effects of random communications bits, which leads to the increased side lobes. For our proposed AFDM-ISAC system, there is a slight loss of image SINR for both cases of $f_d{=}0$ and $f_d{=}0.98\Delta_f$ compared with that of the Swerling 0 model, which means that the AFDM-ISAC system can also work in the Swerling model 3.}

\begin{figure}[!htbp]
	\vspace*{-5pt} %留空白，可自己调整
	\centering
	\includegraphics[width=2.2in]{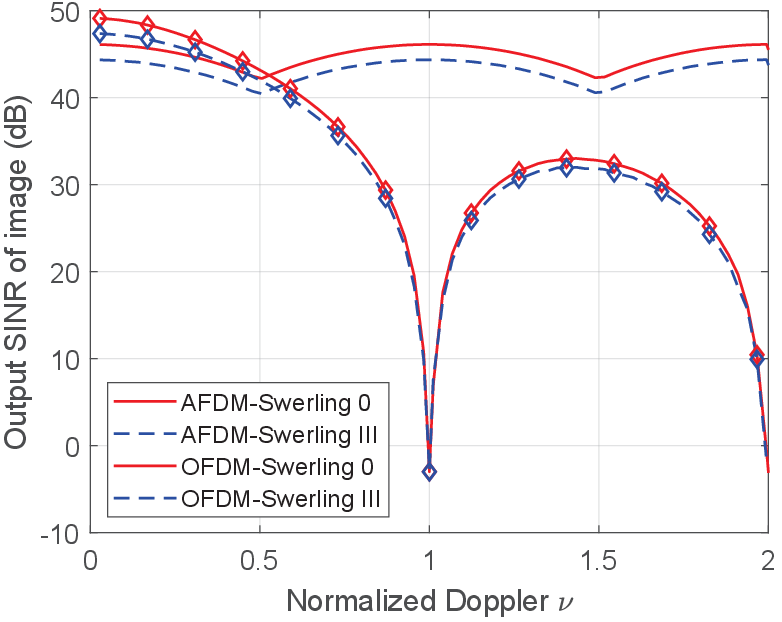}
	\caption{ Output SINR of image versus the normalized Doppler $\nu$ with input echo SNR of 0 dB in both Swerling 0 and Swerling 3 models.}
	\label{fg:Image_SNR_vs_Doppler}
	\vspace*{-5pt} %留空白，可自己调整
\end{figure}

Then, we compare the image SNR versus normalized Doppler shift $\nu$ { under the Swerling 0 and Swerling 3 models}, as shown in Fig. \ref{fg:Image_SNR_vs_Doppler}. { The maximum normalized Doppler shift $\nu _{\max}$ is set to be 2, corresponding to a round-trip velocity of 600 m/s.} Under the Swerling 0 model, it can be seen that in the low Doppler shift region ($\nu < 0.5 $), both systems output similar image SNRs. As Doppler shift increases, the image SNR of the OFDM-ISAC system decreases severely first and then increases. Its minimum SNRs ($<$0 dB) are achieved at the integral $\nu$. However, our AFDM-ISAC system can always output an image SNR greater than 40 dB, which coincides with results in Fig. \ref{fg:Image_SNR_vs_inputSNR}. { Under the Swerling 3 model, both systems have a slight loss of image SINR compared with that of the Swerling 0 model.}

\begin{figure}[!htbp]
	\vspace*{-5pt} %留空白，可自己调整
	\centering
	\includegraphics[width=2.2in]{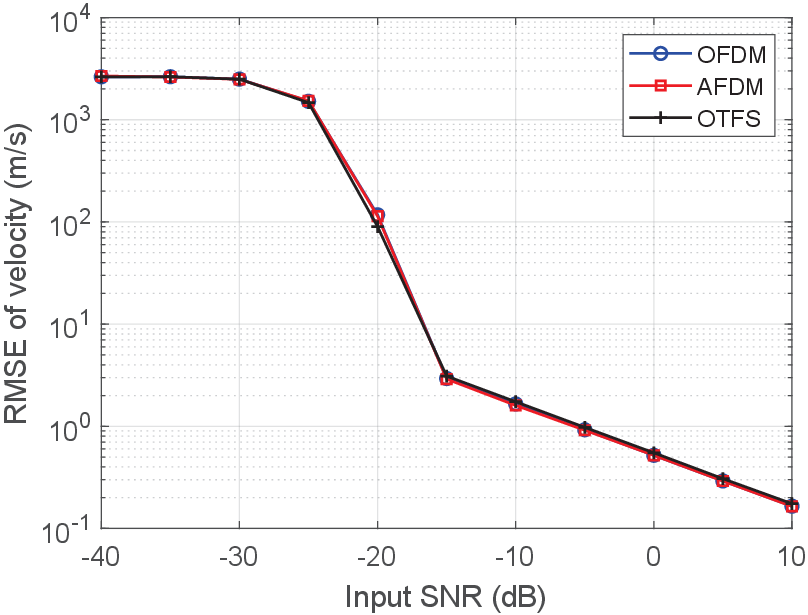}
	\caption{ RMSEs of velocity estimation versus SNR for AFDM-ISAC, OFDM-ISAC and OTFS-ISAC waveforms.
				\label{fg:RMSE_velocity}}
	\vspace*{-5pt} %留空白，可自己调整
\end{figure}
{ Finally, we compare the performances of velocity estimation for AFDM-ISAC, OFDM-ISAC and OTFS-ISAC waveforms. The root mean squared errors (RMSEs) of velocity estimation for three waveforms are shown in Fig. \ref{fg:RMSE_velocity}. The waveform parameters are set following \cite{gaudio2020effectiveness}, where for OFDM and OTFS, $f_c = 5.89$ GHz, $B = 10$ MHz, $\Delta _f = 156.25$ kHz, the number of subcarriers $N_c = 64$, the number of symbols per ISAC frame $N_{sym} = 50$. To make a fair comparison with OFDM and OTFS, we set the number of subcarriers $N_c = 64 \times 50 = 2560$ and the number of symbols per ISAC frame $N_{sym} = 1$ for AFDM such that AFDM has the same resources of bandwidth and time interval as OFDM and OTFS. The maximum likelihood (ML) estimator is used to estimate the velocity for all waveforms\cite{gaudio2020effectiveness}. We can see that AFDM can provide as accurate performance of velocity estimation as OFDM and OTFS. }

\section{Conclusion}

This paper proposed an AFDM-ISAC system. We introduced two new metrics, i.e., SSE and SOP, and designed a metric set for the AFDM-ISAC system.
After that, the analytical relationship between metrics and AFDM waveform parameters was derived, and the trade-offs between S$\&$C performances were analytically investigated for the AFDM-ISAC system. Finally, an estimation method for our AFDM-ISAC system was proposed to estimate delay and Doppler in the AFT-Doppler domain. Numerical results verified that our proposed AFDM-ISAC system significantly enlarged maximum tolerable delay/Doppler and possessed good spectral efficiency and PSLR in high-mobility scenarios.

\section*{APPENDIX}

\appendices
\section{Proof of Lemma 1}\label{proof_Lemma1}
\vspace*{-10pt} %留空白，可自己调整
\small
\begin{align} \label{eq:MI_est_delay}
&{I\left(\tau ^{n,k};{\hat \tau }^{n,k}| U^{n,k},{\hat U}^{n,k}\right)}  \nonumber \\
& = H\left( \tau ^{n,k} | U^{n,k},{\hat U}^{n,k} \right) - H\left( \tau ^{n,k} |{\hat \tau }^{n,k}, U^{n,k},{\hat U}^{n,k} \right),
\end{align} \normalsize
where $H\left( \tau ^{n,k} | U^{n,k},{\hat U}^{n,k} \right)$ and $H\left( \tau ^{n,k} |{\hat \tau }^{n,k}, U^{n,k},{\hat U}^{n,k} \right)$ are given by (\ref{eq:proof_info_est_delay}) and (\ref{eq:proof_info_est_delay2}) on the next page, respectively. Both equality (a) in (\ref{eq:proof_info_est_delay}) and equality (b) in (\ref{eq:proof_info_est_delay2}) come from $f\left( {{\tau ^{n,k}}|{U^{n,k}},{{\hat U}^{n,k}}} \right) = f\left( {{\tau ^{n,k}}|{U^{n,k}}} \right)$, since the distribution of ${\tau ^{n,k}}$ is not affected by ${{\hat U}^{n,k}}$, after ${U^{n,k}}$ is given. 

\begin{figure*}[htbp]
	\vspace*{-10pt} %留空白，可自己调整
\small
	\begin{flalign}\label{eq:proof_info_est_delay}
		&\ H\left( {{\tau ^{n,k}}|{U^{n,k}},{{\hat U}^{n,k}}} \right) =  - \sum\limits_{{U^{n,k}}} {\sum\limits_{{{\hat U}^{n,k}}} {\int_{ - \infty }^\infty  {f\left({U^{n,k}}, {{\tau ^{n,k}},{{\hat U}^{n,k}}} \right)\log f\left( {{\tau ^{n,k}}|{U^{n,k}},{{\hat U}^{n,k}}} \right)d{\tau ^{n,k}}} } } & \nonumber\\
		&\ \quad\quad\quad\quad\quad\quad\quad\quad\quad \stackrel{a}=  - \sum\limits_{{U^{n,k}}} {\sum\limits_{{{\hat U}^{n,k}}} {\int_{ - \infty }^\infty  {{\Pr}\left( {{U^{n,k}}} \right){\Pr}\left( {{{\hat U}^{n,k}}|{U^{n,k}}} \right)f\left( {{\tau ^{n,k}}|{U^{n,k}}} \right)\log f\left( {{\tau ^{n,k}}|{U^{n,k}}} \right)d{\tau ^{n,k}}} } } & \nonumber \\
		&\ \quad\quad\quad\quad\quad\quad\quad\quad\quad = \gamma \log {\Delta _\tau } - \left( {1 - \gamma } \right)\log \delta \left( 0 \right). & 
\end{flalign}\normalsize
\vspace*{-10pt} %留空白，可自己调整
\end{figure*}
\begin{figure*}[htbp]
	\vspace*{-10pt} %留空白，可自己调整
	\small
		\begin{align}\label{eq:proof_info_est_delay2}
			&H\left( {{\tau ^{n,k}}|{{\hat \tau }^{n,k}},{U^{n,k}},{{\hat U}^{n,k}}} \right) =  - \sum\limits_{{U^{n,k}}} {\sum\limits_{{{\hat U}^{n,k}}} {\sum\limits_{{{\hat \tau}^{n,k}}}  {\int_{ - \infty }^\infty  {f\left( {{\tau ^{n,k}},{{\hat \tau }^{n,k}},{U^{n,k}},{{\hat U}^{n,k}}} \right)\log f\left( {{\tau ^{n,k}}|{{\hat \tau }^{n,k}},{U^{n,k}},{{\hat U}^{n,k}}} \right)d{\tau ^{n,k}}} } } } \nonumber \\
			&\stackrel{b}=  - \sum\limits_{{U^{n,k}}} {\sum\limits_{{{\hat U}^{n,k}}} {\sum\limits_{{{\hat \tau}^{n,k}}}  {\int_{ - \infty }^\infty  {{\Pr}\left( {{U^{n,k}}} \right){\Pr}\left( {{{\hat U}^{n,k}}|{U^{n,k}}} \right)f\left( {{\tau ^{n,k}}|{U^{n,k}}} \right){\Pr}\left( {{{\hat \tau }^{n,k}}|{\tau ^{n,k}},{U^{n,k}},{{\hat U}^{n,k}}} \right)\log f\left( {{\tau ^{n,k}}|{{\hat \tau }^{n,k}},{U^{n,k}},{{\hat U}^{n,k}}} \right)d{\tau ^{n,k}}} } } } .
				\end{align}\normalsize
\hrule
\vspace*{-10pt} %留空白，可自己调整
\end{figure*}

When the detector outputs target being absent in the $\left(n,k\right)$-th cell, the delay is estimated as null. Hence, we have
\small
\begin{align}\label{eq:pdf_tau_est}
&\	{\Pr}\left( {{{\hat \tau }^{n,k}}|{\tau ^{n,k}},{U^{n,k}} = 1,{{\hat U}^{n,k}} = 0} \right) \hspace*{-1px}=\hspace*{-1px} \left\{ {\begin{array}{*{20}{c}}
			\hspace*{-5px}{1,}&\hspace*{-7px} {{{\hat \tau }^{n,k}} \in {\phi} ,}\\
			\hspace*{-5px}{0,}&\hspace*{-7px}{otherwise,}
	\end{array}} \right.&
\end{align}\normalsize
\small
	\begin{align}\label{eq:pdf_tau_est2}		
		&\	{\Pr}\left( {{{\hat \tau }^{n,k}}|{\tau ^{n,k}},{U^{n,k}} = 0,{{\hat U}^{n,k}} = 0} \right) \hspace*{-1px}=\hspace*{-1px} \left\{ {\begin{array}{*{20}{c}}
				\hspace*{-5px}{1,}&\hspace*{-7px}{{{\hat \tau }^{n,k}} \in \phi ,}\\
				\hspace*{-5px}{0,}&\hspace*{-7px}{otherwise,}
		\end{array}} \right.&
\end{align}\normalsize
where ${\phi}$ denotes the empty set.
When $U^{n,k}=1$, ${{\hat U}^{n,k}}=1$, and the correct sub-cell has been determined, the probability of ${\Pr}\left( {{{\hat \tau }^{n,k}}|{\tau ^{n,k}},{U^{n,k}} = 1,{{\hat U}^{n,k}} = 1} \right)$ is given by
\begin{equation}
	{\Pr}\left( {{{\hat \tau }^{n,k}}|{\tau ^{n,k}},{U^{n,k}}\hspace*{-1px} =\hspace*{-1px} 1,{{\hat U}^{n,k}} \hspace*{-1px}=\hspace*{-1px} 1} \right) \hspace*{-1px}=\hspace*{-1px} \left\{ {\begin{array}{*{20}{c}}
			\hspace*{-5px}{1,}&\hspace*{-7px}{{{\hat \tau }^{n,k}} \in \mathbb{T}_1,}\\
			\hspace*{-5px}{0,}&\hspace*{-7px}{otherwise,}
	\end{array}} \right.\\
\end{equation}
where $\mathbb{T}_1 = \left[ {{\tau ^{n,k}} - 0.5{\bar D_\tau },{\tau ^{n,k}} + 0.5{\bar D_\tau }} \right)$.
When $U^{n,k}=0$ and ${{\hat U}^{n,k}}=1$, the peak above the threshold is contributed by the noise, and thus the delay corresponding to the peak value is random. As a result, the conditional probability of ${{\hat \tau }^{n,k}}$ is written as
\begin{equation}
	{\Pr}\left( {{{\hat \tau }^{n,k}}|{\tau ^{n,k}},{U^{n,k}} = 0,{{\hat U}^{n,k}} = 1} \right) = \left\{ {\begin{array}{*{20}{c}}
			{\frac{\bar D_{\tau}}{{{\Delta _{\tau} }}},}&{{{\hat \tau }^{n,k}} \in \mathbb{T}_2,}\\
			{0,}&{otherwise,}
	\end{array}} \right.
\end{equation}
where $\mathbb{T}_2 = \left[ {{{\bar \tau }^{n,k}} - 0.5{\Delta _\tau },{{\bar \tau }^{n,k}} + 0.5{\Delta _\tau }} \right)$.
Since in generic $f\left( {x|y,z} \right) = \frac{{f\left( {x|z} \right)f\left( {y|x,z} \right)}}{{f\left( {y|z} \right)}}$, we can get
$V_{\max}$
\begin{align}\label{eq:pdf_tau}
	&f\left( {{\tau ^{n,k}}|{{\hat \tau }^{n,k}},{U^{n,k}} = 1,{{\hat U}^{n,k}} = 1} \right) = \left\{ {\begin{array}{*{20}{c}}
			{\frac{1}{{\bar D_{\tau}}},}&{{\tau ^{n,k}} \in \mathbb{T}_3,} \\
			{0,}&{otherwise,}
	\end{array}} \right. \nonumber\\
	&	f\left( {{\tau ^{n,k}}|{{\hat \tau }^{n,k}},{U^{n,k}} = 0,{{\hat U}^{n,k}} = 1} \right) = \delta \left( {{\tau ^{n,k}}} \right),\nonumber\\
	&	f\left( {{\tau ^{n,k}}|{{\hat \tau }^{n,k}},{U^{n,k}} = 1,{{\hat U}^{n,k}} = 0} \right) = \left\{ {\begin{array}{*{20}{c}}
			{\frac{1}{{{\Delta _\tau }}},}&{{\tau ^{n,k}} \in \mathbb{T}_2 ,}\\
			{0,}&{otherwise,}
	\end{array}} \right.\nonumber\\
	&	f\left( {{\tau ^{n,k}}|{{\hat \tau }^{n,k}},{U^{n,k}} = 0,{{\hat U}^{n,k}} = 0} \right) = \delta \left( {{\tau ^{n,k}}} \right),
\end{align}\normalsize
where $\mathbb{T}_3 = \left[ {{{\hat \tau }^{n,k}} - 0.5{\bar D_\tau },{{\hat \tau }^{n,k}} + 0.5{\bar D_\tau }} \right)$.
Substituting (\ref{eq:pdf_tau_est}) $\sim$ (\ref{eq:pdf_tau}) to (\ref{eq:proof_info_est_delay2}), we can get
\begin{align} \label{eq:re_info_est_delay}
	& H\left( {{\tau ^{n,k}}|{{\hat \tau }^{n,k}},{U^{n,k}},{{\hat U}^{n,k}}} \right) = \gamma {P_{D}^{n,k}}\log {\bar D_\tau } + \gamma \log {\Delta _\tau } \nonumber \\
	&\qquad \qquad \qquad - \gamma {P_{D}^{n,k}}\log {\Delta _\tau } - \left( {1 - \gamma } \right)\log \delta \left( 0 \right).
\end{align}
Substituting (\ref{eq:proof_info_est_delay}) and (\ref{eq:re_info_est_delay}) to (\ref{eq:MI_est_delay}), ${I\left(\tau ^{n,k};{\hat \tau }^{n,k}| U^{n,k},{\hat U}^{n,k}\right)}$ can be written as (\ref{eq:info_est_delay}). Lemma 1 is proved.

\section{Proof of Corollary 1}\label{proof_Theorem1} 

Let $x_1 = {\frac{{\tau ^{U} - {\tau ^{n,k}}}}{{{\sigma _{\tau }^{n,k}}}}}$, and $x_2 = {\frac{{\tau ^{L} - {\tau ^{n,k}}}}{{{\sigma _{\tau }^{n,k}}}}}$.
$x_1 - x_2 = {\frac{\bar D_{\tau}}{{{\sigma _{\tau }^{n,k}}}}}\left(1+\left\lceil {{{\left( {{\tau ^{n,k}} - {{\bar \tau }^{n,k}}} \right)} \mathord{\left/
			{\vphantom {{\left( {{\tau ^{n,k}} - {{\bar \tau }^{n,k}}} \right)} {{\bar D_\tau }}}} \right.
			\kern-\nulldelimiterspace} {{\bar D_\tau }}}} \right\rceil - \left\lfloor {{{\left( {{\tau ^{n,k}} - {{\bar \tau }^{n,k}}} \right)} \mathord{\left/
			{\vphantom {{\left( {{\tau ^{n,k}} - {{\bar \tau }^{n,k}}} \right)} {{\bar D_\tau }}}} \right.
			\kern-\nulldelimiterspace} {{\bar D_\tau }}}} \right\rfloor\right)$.
When ${\tau ^{n,k}} = \left\lfloor {{{\left( {{\tau ^{n,k}} - {{\bar \tau }^{n,k}}} \right)} \mathord{\left/
			{\vphantom {{\left( {{\tau ^{n,k}} - {{\bar \tau }^{n,k}}} \right)} {{\bar D_\tau }}}} \right.
			\kern-\nulldelimiterspace} {{\bar D_\tau }}}} \right\rfloor {\bar D_\tau } + {{\bar \tau }^{n,k}}$, we can get ${\tau ^{L}}-{\tau ^{n,k}}=-0.5\bar D_{\tau}$ and ${\tau ^{U}}-{\tau ^{n,k}}=0.5\bar D_{\tau}$. At this time, $x_1 - x_2$ reaches its minimum value ${\frac{\bar D_{\tau}}{{{\sigma _{\tau }^{n,k}}}}}$. Consequently, $Q\left( {\frac{{\tau ^{L} - {\tau ^{n,k}}}}{{{\sigma _{\tau }^{n,k}}}}} \right) - Q\left( {\frac{{\tau ^{U} - {\tau ^{n,k}}}}{{{\sigma _{\tau }^{n,k}}}}} \right)$ has its minimum value $1-2Q\left( {\frac{\bar D_{\tau}}{{2{\sigma _{\tau }^{n,k}}}}} \right)$. Similarly, when ${f_{d}^{n,k}} = \left\lfloor {{{\left( {{f_{d}^{n,k}} - {{\bar f}_{d}^{n,k}}} \right)} \mathord{\left/
			{\vphantom {{\left( {{f_{d}^{n,k}} - {{\bar f}_{d}^{n,k}}} \right)} {{\bar D_{{f_d}}}}}} \right.
			\kern-\nulldelimiterspace} {{\bar D_{{f_d}}}}}} \right\rfloor {\bar D_{{f_d}}} + {{\bar f}_{d}^{n,k}}$, $Q\left( {\frac{{f_{d}^L - {f_{d}^{n,k}}}}{{{\sigma _{{f_d}}^{n,k}}}}} \right) - Q\left( {\frac{{f_{d}^{U} - {f_{d}^{n,k}}}}{{{\sigma _{{f_d}}^{n,k}}}}} \right)$ has its minimum value $1-2Q\left( {\frac{\bar D_{f_d}}{{2{\sigma _{f_{d}}^{n,k}}}}} \right)$. Therefore, SOP has its maximum value $P_{e,sen}^{\max }$ as shown in (\ref{eq:SOP_max}).
		
When ${\tau ^{n,k}} \neq \left\lfloor {{{\left( {{\tau ^{n,k}} - {{\bar \tau }^{n,k}}} \right)} \mathord{\left/
			{\vphantom {{\left( {{\tau ^{n,k}} - {{\bar \tau }^{n,k}}} \right)} {{\bar D_\tau }}}} \right.
			\kern-\nulldelimiterspace} {{\bar D_\tau }}}} \right\rfloor {\bar D_\tau } + {{\bar \tau }^{n,k}}$, we can get $x_1 - x_2 = {\frac{2D_{\tau}}{{{\sigma _{\tau }^{n,k}}}}}$ holding for any ${\tau ^{n,k}}$. At this time, $Q\left( {\frac{{\tau ^{L} - {\tau ^{n,k}}}}{{{\sigma _{\tau }^{n,k}}}}} \right) - Q\left( {\frac{{\tau ^{U} - {\tau ^{n,k}}}}{{{\sigma _{\tau }^{n,k}}}}} \right)$ can achieve its maximum value $1-2Q\left( {\frac{\bar D_{\tau}}{{{\sigma _{\tau }^{n,k}}}}} \right)$, when ${\tau ^{n,k}} = \left\lceil {{{\left( {{\tau ^{n,k}} - {{\bar \tau }^{n,k}}} \right)} \mathord{\left/
			{\vphantom {{\left( {{\tau ^{n,k}} - {{\bar \tau }^{n,k}}} \right)} {{\bar D_\tau }}}} \right.
			\kern-\nulldelimiterspace} {{\bar D_\tau }}}} \right\rceil {\bar D_\tau } + {{\bar \tau }^{n,k}} - 0.5{\bar D_\tau }$.	Similarly, when ${f_{d}^{n,k}} = \left\lceil {{{\left( {{f_{d}^{n,k}} - {{\bar f}_{d}^{n,k}}} \right)} \mathord{\left/
			{\vphantom {{\left( {{f_{d}^{n,k}} - {{\bar f}_{d}^{n,k}}} \right)} {{\bar D_{{f_d}}}}}} \right.
			\kern-\nulldelimiterspace} {{\bar D_{{f_d}}}}}} \right\rceil {\bar D_{{f_d}}} + {{\bar f}_{d}^{n,k}} - 0.5{\bar D_{{f_d}}}$, $Q\left( {\frac{{f_{d}^L - {f_{d}^{n,k}}}}{{{\sigma _{{f_d}}^{n,k}}}}} \right) - Q\left( {\frac{{f_{d}^{U} - {f_{d}^{n,k}}}}{{{\sigma _{{f_d}}^{n,k}}}}} \right)$ has its maximum value $1-2Q\left( {\frac{\bar D_{f_d}}{{{\sigma _{f_{d}}^{n,k}}}}} \right)$. Therefore, SOP has its minimum value $P_{e,sen}^{\min }$ as shown in (\ref{eq:SOP_min}). Corollary 1 is proved.

\small
\bibliographystyle{IEEEbib}
\bibliography{IEEEabrv,IEEE_JRCJ_ref}

\end{document}